\documentclass[12pt]{article}

\usepackage{amsmath,amssymb,amsthm}
\usepackage{graphicx}
\usepackage{subfigure}
\usepackage{fullpage}
\usepackage{verbatim}
\usepackage[dvips]{color}
\usepackage[round]{natbib}

\newtheorem{theorem}{Theorem}[section]
\newtheorem{proposition}[theorem]{Proposition}
\newtheorem{lemma}[theorem]{Lemma}
\newtheorem{corollary}[theorem]{Corollary}

\theoremstyle{remark}
\newtheorem{definition}[theorem]{Definition}
\newtheorem{example}[theorem]{Example}
\newtheorem{remark}[theorem]{Remark}
\newtheorem{examplec}{Example}

\def\mbf#1{\mathchoice{\hbox{\boldmath $\displaystyle #1$}}
        {\hbox{\boldmath $\textstyle #1$}}
        {\hbox{\boldmath $\scriptstyle #1$}}
        {\hbox{\boldmath $\scriptscriptstyle #1$}}}

\newcommand{\X}{{\mbf X}}
\newcommand{\Y}{{\mbf Y}}
\newcommand{\K}{{\mbf K}}
\newcommand{\sZb}{\boldsymbol{\mathcal Z}}

\newcommand{\R}{{\mathbb R}}
\newcommand{\HH}{{\mathbb H}}
\newcommand{\E}{{\mathbf E}}
\newcommand{\sA}{\mathcal{A}}
\newcommand{\sM}{\mathcal{M}}

\newcommand{\sZ}{\mathcal{Z}}
\newcommand{\gc}{{\mathbb G}}
\newcommand{\Lip}{\mathrm{Lip}}
\newcommand{\lipc}[1]{\mathsf{L}_{#1}}
\newcommand{\bplus}{+}

\renewcommand{\P}{\mathbf{P}}
\newcommand{\Prob}[1]{\P\{#1\}}
\newcommand{\one}{\mathbf{1}}
\newcommand{\eps}{\varepsilon}
\newcommand{\rhos}{\rho_{\mathrm{s}}}
\newcommand{\rhonot}{\rho_{\mathrm{s},0}}
\newcommand{\maxcor}{\Phi}

\newcommand{\salg}{\mathfrak{F}}
\newcommand{\sel}{L}
\newcommand{\dha}{{\mathfrak{d}}_{\mathrm H}}

\newcommand{\risk}{r}
\newcommand{\vecrisk}{\mathbf{r}}
\newcommand{\numrisk}{\mathsf{r}}

\newcommand{\rhogen}{\rho}
\newcommand{\rhosel}{\rho_{\sZb}}
\newcommand{\rhoselnot}{\rho_{\sZb_0}}

\DeclareMathOperator{\cl}{cl}

\DeclareMathOperator{\essinf}{essinf}
\DeclareMathOperator{\esssup}{esssup}
\DeclareMathOperator{\ES}{ES}
\DeclareMathOperator{\VaR}{VaR}

\newlength{\querylen}
\usepackage{fancybox}

\begin{document}

\begin{center}
  \large\bf MULTIVARIATE RISK MEASURES: A CONSTRUCTIVE APPROACH BASED ON
  SELECTIONS
  \footnote{The authors are grateful to Leonid Hanin for his advice on duality
results in Lipschitz spaces. IM acknowledges the hospitality of the
University Carlos III de Madrid. IM has benefited from discussions
with Qiyu Li, Michael Schmutz and Irina Sikharulidze at various stages
of this work. The second version of this preprint was greatly inspired
by insightful comments of Birgit Rudloff concerning her recent work on
multivariate risk measures. The authors are grateful to the Associate
Editor and the referees for thoughtful comments and encouragement that
led to a greatly improved paper.\\ This version corrects Lemma 7.1 and
Lemma 7.2 in Section 7.}
\end{center}

\begin{center}
  \large 
  \textsc{Ignacio Cascos}\\
  \textit{Universidad Carlos III de Madrid}
\end{center}

\begin{center}
  \large
  \textsc{Ilya Molchanov}\\
  \textit{University of Bern}
  \footnote{IC supported by the Spanish Ministry of
    Science and Innovation Grants No. MTM2011-22993 and
    ECO2011-25706. IM supported by the
    Chair of Excellence Programme of the Universidad Carlos III de
    Madrid and Banco Santander and the Swiss National Foundation Grant
    No. 200021-137527.}
  \footnote{Address correspondence to Ilya Molchanov, Institute of
    Mathematical Statistics and Actuarial Science, University of Bern,
    Sidlerstrasse 5, 3012 Bern, Switzerland; e-mail:
    ilya.molchanov@stat.unibe.ch.} 
\end{center}



\begin{quotation}
  \small
  Since risky positions in multivariate portfolios can be offset by
  various choices of capital requirements that depend on the exchange
  rules and related transaction costs, it is natural to assume that
  the risk measures of random vectors are set-valued. Furthermore, it
  is reasonable to include the exchange rules in the argument of the
  risk measure and so consider risk measures of set-valued
  portfolios. This situation includes the classical Kabanov's
  transaction costs model, where the set-valued portfolio is given by
  the sum of a random vector and an exchange cone, but also a number
  of further cases of additional liquidity constraints.

  We suggest a definition of the risk measure based on calling a
  set-valued portfolio acceptable if it possesses a selection with all
  individually acceptable marginals.  The obtained selection risk
  measure is coherent (or convex), law invariant and has values being
  upper convex closed sets. We describe the dual representation of the
  selection risk measure and suggest efficient ways of approximating
  it from below and from above. In the case of Kabanov's exchange cone
  model, it is shown how the selection risk measure relates to the
  set-valued risk measures considered by Kulikov (2008), Hamel and
  Heyde (2010), and Hamel, Heyde and Rudloff (2013).

\end{quotation}

\section{Introduction}
\label{sec:introduction}

Since the seminal papers by \cite*{art:del:eber:99} and \cite{delb02},
most studies of risk measures deal with the univariate
case, where the gains or liabilities are expressed by a random
variable. We refer to \cite{foel:sch04} and \cite*{mcn:frey:emb05} for
a thorough treatment of univariate risk measures and to
\cite{acc:pen11} for a recent survey of the dynamic univariate
setting.  

Multiasset portfolios in practice are often represented by their total
monetary value in a fixed currency with the subsequent calculation of
univariate risk measures that can be used to determine the overall
capital requirements. The main emphasis is put on the dependency
structure of the various components of the portfolio, see
\cite{bur:rues06,emb:puc06}. Numerical risk measures for a
multivariate portfolio $X$ have been also studied in
\cite*{ekel:gal:hen12} and \cite{rues06}. The key idea is to consider
the expected scalar product of $(-X)$ with a random vector $Z$ and
take supremum over all random vectors that share the same distribution
with $X$ and possibly over a family of random vectors
$Z$. \cite*{far:koc:mun13} defined the scalar risk as the infimum of
the payoff associated with a random vector that being added to the
portfolio renders it acceptable. A vector-valued variant of the
value-at-risk has been suggested in \cite{cous:dib13}.

\bigskip
However, in many natural applications it is necessary to assess the
risk of a vector $X$ in $\R^d$ whose components represent different
currencies or gains from various business lines, where profits/losses
from one line or currency cannot be used directly to offset the
position in a different one. Even in the absence of transaction costs,
the exchange rates fluctuate and so may influence the overall risk
assessment. Also the regulatory requirements may be very different for
different lines (e.g. in the case of several states within the same
currency area), and moving assets may be subject to transaction costs,
taxes or other restrictions. For such cases, it is important to
determine the necessary reserves that should be allocated in each line
(currency or component) of $X$ in order to make the overall position
acceptable. The simplest solution would be to treat each component
separately and allocate reserves accordingly, which is not in the
interest of (financial) agents who might want to use profits from
one line to compensate for eventual losses in other ones.  Thus, in
addition of assessing the risk of the original vector $X$, one can
also evaluate the risk of any other portfolio that may be obtained
from $X$ by allowed transactions. In view of this, it is natural to
assume that the acceptability may be achieved by several (and not
directly comparable) choices of capital requirements that form a set
of possible values for the risk measure. This suggests the idea of
working with \emph{set-valued} risk measures.

Since the first work on multivariate risk measures
\citep*{jouin:med:touz04}, by now it is accepted that multiasset risk
measures can be naturally considered as taking values in the space of
sets, see
\cite*{bent:lep13,cas:mol07,ham:hey10,ham:hey:rud11,kul08}. The risk
measures of random vectors are mostly considered in relation to
Kabanov's transaction costs model, whose main ingredient is a cone
$\K$ of portfolios available at price zero at the chosen time horizon
(also called the \emph{exchange cone}), while the central symmetric
variant of $\K$ is a solvency cone. If $X$ is the terminal gain, then
each random vector with values in $X+\K$ is possible to obtain by
converting $X$ following the rules determined by $\K$. In other words,
instead of measuring the risk of $X$ we consider the whole family of
random vectors taking values in $X+\K$. In relation to this, note that
families of random vectors representing attainable gains are often
considered in the financial studies of transaction costs models, see
e.g. \cite{schach01}.

\bigskip
The set-valued portfolios framework can be related to the classical
setting by replacing a univariate random gain $X$ with half-line
$(-\infty,X]$ and measuring risks of all random variables dominated by
$X$. The monotonicity property of a chosen risk measure $\risk$
implies that all these risks build the set
$\rho(X)=[\risk(X),\infty)$. If $\risk$ is subadditive, then
$\rho(X+Y)\supset \rho(X)+\rho(Y)$. Furthermore, $\risk(X)\leq 0$ if
and only if $\rho(X)$ contains the origin. While in the univariate
case this construction leads to half-lines, in the multivariate
situation it naturally gives rise to so-called upper convex sets, see
\cite{ham:hey10,ham:hey:rud11,kul08}. A portfolio is acceptable if its
risk measure contains the origin and the value of the risk measure is
the set of all $a\in\R^d$ such that the portfolio becomes acceptable
if capital $a$ is added to it.

Note that in the setting of real-valued risk measures adapted in
\cite{ekel:gal:hen12}, the family of all $a\in\R^d$ that make $X+a$
acceptable is a half-space, which apparently only partially reflects
the nature of cone-based transaction costs
models. \cite{far:koc:mun13} establish relation between families of
real-valued risks and set-valued risks from \cite{ham:hey:rud11}. The
setting of Riesz spaces (partially ordered linear spaces), in
particular Fr\'echet lattices and Orlicz spaces, has become already
common in the theory of risk measures, see
\cite{biag:frit08,cher:li09}. However, these spaces are mostly used to
describe the arguments of risk measures whose values belong to the
(extended) real line. Furthermore, the space of sets is no longer a
Riesz space --- while the addition is well defined, the matching
subtraction does not exist. The recent study of risk preferences
\citep{drap:kup13} also concentrates on the case of vector spaces for
arguments of risk measures.

The dual representation for risk measures of random vectors in the
case of a deterministic exchange cone is obtained in \cite{ham:hey10}
and for the random case in \cite{ham:hey:rud11}, see also \cite{kul08}
who considers both deterministic and random exchange cones. However in
the case of a random exchange cone, it does not produce law-invariant risk
measures --- the risk measure in
\cite{ham:hey10,ham:hey:rud11,ham:rud:yan13,kul08} is defined as a
function of a random vector $X$ representing the gain, while
identically distributed gains might exhibit different properties in
relation to the random exchange cone.  Although the dual
representations from \cite{ham:hey10,ham:hey:rud11,kul08} are general,
they are rather difficult to use in order to calculate risks for given
portfolios, since they are given as intersections of half-spaces
determined by a rather rich family of random vectors from the dual
space. Recent advances in vector optimisation have led to a
substantial progress in computation of set-valued risk measures, see
\cite{ham:rud:yan13}.
However, the dual approximation also does not explicitly yield the
relevant trading (or exchange) strategy that determines transactions
suitable to compensate for risks. The construction of set-valued risk
measures from \cite{cas:mol07} is based on the concept of the
depth-trimmed region, and their values are easy to calculate
numerically or analytically, but it only applies for deterministic
exchange cones and often results in marginalised risks (so that the
risk measure is a translate of the solvency cone).

\bigskip
In order to come up with a law invariant risk measure and also cover
the case of random exchange cones, we assume that the argument of a
risk measure is a \emph{random closed set} that consists of all
attainable portfolios. This random set may be the sum $X+\K$ of a
random vector $X$ and the exchange cone $\K$ (which has been the most
important example so far) or may be defined otherwise. For instance,
if only a linear space of portfolios $M$ is available for
compensation, then the set of attainable portfolios is the
intersection of $X+\K$ with $M$. 
In any such case we speak about a set-valued portfolio $\X$. This
guiding idea makes it possible to work out the law-invariance property
of risk measures and naturally arrive at set-valued risks.

\bigskip
In this paper we suggest a rather simple and intuitive way to measure
risks for set-valued portfolios based on considering the family of all
terminal gains that may be attained after some exchanges are
performed. The crucial step is to consider all random vectors taking
values in a random set $\X$ (\emph{selections} of $\X$) as possible
gains and regard the random set acceptable if it possesses a selection
with all acceptable components. In view of this, we do not only
determine the necessary capital reserves, but also the way of
converting the terminal value of the portfolio into an acceptable
one. 

In the case of exchange cones, we relate our construction to the dual
representation from \cite{ham:hey10} and \cite{kul08}.  Throughout the
paper we concentrate on the coherent case and one-period setting, but
occasionally comment on non-coherent
generalisations. \cite{fein:rud14c} thoroughly analyse and compare
various approaches, including one from this paper, in view of defining
multiperiod set-valued risks, see also
\cite{fein:rud14fs}. 

\begin{example}
  \label{ex:intro}
  Let $X=(X_1,X_2)$ represent terminal gains on two business lines
  expressed in two different currencies. Assume that $X_1$ and $X_2$
  are i.i.d. normally distributed with mean $0.5$ and variance
  $1$. Assume that the exchanges between currencies are free from
  transaction costs with the initial exchange rate $\pi_0=1.5$ (number
  of units of the second currency to buy one unit of the first one),
  the terminal exchange rate $\pi$ is lognormal with mean $1.5$ and
  volatility $0.4$, and $\pi$ is independent of $X$. The set-valued
  portfolio $\X$ is a half-plane with the boundary passing through $X$
  and normal $(\pi,1)$.
  
  Assume that necessary capital reserves are determined using the
  expected shortfall $\ES_{0.05}$ at level $0.05$, see
  \cite{acer:tas02}.  If compensation between business lines is not
  allowed, the necessary capital reserves at time zero are given by
  $\ES_{0.05}(X_1)+\pi_0^{-1}\ES_{0.05}(X_2)\approx 2.6045$ (all
  numbers are given in the units of the first currency). If the
  terminal gains are transferred to one currency, then the needed
  reserves are given by $\ES_{0.05}(X_1+\pi^{-1}X_2)\approx 1.801$ and
  $\pi_0^{-1}\ES_{0.05}(\pi X_1+X_2)\approx 1.784$ respectively in the
  case of transfers to the first and the second currency.  These
  values correspond to evaluating the risks of selections of $\X$
  located at the points of intersection of the boundary of $\X$ with
  coordinate axes.

  However, it is possible to choose a selection that further reduces the
  required capital requirements. After transferring to the second
  asset $(X_1-\pi X_2)/(1+\pi^2)$ units of the first currency, we
  arrive at the selection of $\X$ given by
  \begin{equation}
    \label{eq:sel-p}
    \xi=(\xi_1,\xi_2)=\left(\frac{\pi(X_1\pi+X_2)}{1+\pi^2},
    \frac{X_1\pi+X_2}{1+\pi^2}\right)
  \end{equation}
  obtained by projecting the origin onto the boundary of $\X$. In this
  case the needed reserves are
  $\ES_{0.05}(\xi_1)+\pi_0^{-1}\ES_{0.05}(\xi_2)\approx 1.661$.
  The situation of random frictionless exchanges is also considered
  analytically in Example~\ref{ex:half-random} and numerically in
  Examples~\ref{ex:two-currencies} and \ref{ex:restricted-liquidity}. 

  Assume that now liquidity restrictions are imposed meaning that at
  most one unit of each currency may be obtained after conversion from
  the other. This framework corresponds to dealing with a non-conical
  set-valued portfolio $\Y=\X\cap (X+(1,1)+\R_-^2)$. A reasonable
  strategy would be to use selection $\xi$ from \eqref{eq:sel-p} if
  $(\pi X_2- X_1)/(1+\pi^2)\leq 1$ and $\pi(X_1-\pi
  X_2)/(1+\pi^2)\leq 1$ and otherwise choose the nearest point to
  $\xi$ from the extreme points of $\Y$. The corresponding selection
  is given by
  \begin{equation}
    \label{eq:eta-intro}
    \eta=(\eta_1,\eta_2)=
    \begin{cases}
      (X_1+1,X_2-\pi) & \textrm{if }\pi X_2-X_1>1+\pi^2\,, \\
      (X_1-\pi^{-1},X_2+1) & \textrm{if }\pi X_1-\pi^2X_2>1+\pi^2\,, \\
      \xi & \textrm{otherwise}\,,
    \end{cases}
  \end{equation}
  and the needed reserves are
  $\ES_{0.05}(\eta_1)+\pi_0^{-1}\ES_{0.05}(\eta_2)\approx 1.735$,
  which is higher than those corresponding to the choice of $\xi$ in
  view of imposed liquidity restrictions. 
\end{example}

\bigskip
Before describing the structure of the paper, we would like to point
out that our approach is \emph{constructive} in the sense that instead
of starting with an axiomatic definition of a set-valued risk measure
we explicitly construct one based on selections of a random portfolio
and univariate marginal risk measures. Then we show that the
constructed risk measure indeed satisfies the desired properties of
set-valued risk measures, in particular, the coherency and the Fatou
properties, instead of imposing them. We show how to approximate the
values of the risk measure from below (which is the aim of the market
regulator) and from above (as the agent would aim to do). The
suggested bounds provide a feasible alternative to exact calculations
of risk. 
Furthermore, the computational burden is passed to the agent who aims
to increase the family of selections in order to obtain a tighter
approximation from above and so reduce the capital requirements, quite
differently to the dual constructions of \cite{ham:hey:rud11} and
\cite{kul08}, where the market regulator faces the task of making the
acceptance criterion more stringent by approximating from below the
exact value of the risk measure. It should be noted our approach
constitutes just one possible way to construct multivariate risk
measures (and the corresponding acceptance sets) that satisfy the
axioms of set-valued coherent risk measures from \cite{ham:hey:rud11}.

\bigskip
Section~\ref{sec:acceptance-sets} introduces the concept of set-valued
portfolios and the definition of set-valued risk measures for
set-valued portfolios adapted from \cite{ham:hey:rud11},
where the conical setting was
considered. Section~\ref{sec:select-risk-meas} defines the
\emph{selection risk measure}, which relies on $d$ univariate risk
measures applied to the components of selections for a set-valued
portfolio. In particular, the coherency of the selection risk measure
is established in Theorem~\ref{thr:rhos}. While throughout the paper
we work with coherent risk measures defined on $L^p$ spaces with
$p\in[1,\infty]$, the construction can be also based on convex
non-coherent and non-convex univariate risk measures, so that it
yields their non-coherent set-valued analogues, such as the
value-at-risk.

Section~\ref{sec:bounds-set-valued} derives lower and upper bounds for
risk measures. Section~\ref{sec:rand-exch-cones} is devoted to the
setting of \emph{exchange cones} (or conical market models) that has
been in the centre of attention in all other works on multiasset
risks. It is shown that, for the exchange cones setting, the lower
bound corresponds to the dual representation of risk measures from
\cite{ham:hey10,ham:hey:rud11} and \cite{kul08}. For deterministic
exchange cones the bounds become even simpler and in the case of
comonotonic portfolios the risk measure admits an easy expression.

We briefly comment on scalarisation issues in
Section~\ref{sec:numer-risk-meas}, i.e. explain relationships to
univariate risk measures constructed for set-valued portfolios, which
in the case of a deterministic exchange cone are related to those
considered in \cite{ekel:gal:hen12,rues06}  and \cite{far:koc:mun13}.

Section~\ref{sec:cont-dual-repr} establishes the dual representation
of the selection risk measures. While the idea is to handle set-valued
portfolios through their support functions, the key difficulty
consists in dealing with possibly unbounded values of the support
functions. For this, we introduce the Lipschitz space of random sets
and specify the weak-star convergence in this space in order to come
up with a general dual representation for set-valued risk measures
with the Fatou property. Theorems~\ref{thr:fp-rhos},~\ref{thr:fp},
and~\ref{thr:fatou-det-K} establish the Fatou property of the selection
risk measure under some conditions 
and so yield the closedness of its values and the
validity of the dual representation. In the deterministic
exchange cone model and for random exchange cones with
$p\in(1,\infty)$, the selection risk measure has the same dual
representation as in \cite{ham:hey10,kul08}.

Section~\ref{sec:comp-appr-risk} presents several numerical examples
of set-valued risk measures covering the exchange cone setting, the
frictionless case, and liquidity restrictions. The algorithms used to
approximate risk measures are transparent and easy to implement in
comparison with a considerably more sophisticated set-optimisation
approach from \cite{loeh11} used in \cite{ham:rud:yan13} in order to
come up with exact values of set-valued risk measures in conical
models. A particular computational advantage is due to the use of the
primal representation of selection risk measures in order to compute upper
bounds, while utilising the dual representation to arrive at lower
bounds.

\section{Set-valued portfolios and risk measures}
\label{sec:acceptance-sets}

\subsection{Operations with sets}
\label{sec:operations-with-sets}

In order to handle set-valued portfolios, we need to define
several important operations with sets in $\R^d$. The closure
of a set $M$ is denoted by $\cl(M)$.  Further,
\begin{displaymath}
  \check{M}=\{-x:\; x\in M\}
\end{displaymath}
denotes the centrally symmetric set to $M$.  The sum $M+L$ of two
(deterministic) sets $M$ and $L$ in a linear space is defined as the
set $\{x+y:\; x\in M,\; y\in L\}$. If one of the summands is compact
and the other is closed, the set of pairwise sums is also closed. In
particular, the sum $x+M$ of a point and a set is given by $\{x+y:\;
y\in M\}$. For instance, $x+\R_-^d$ is the set of points dominated by
$x$, where $\R_-^d=(-\infty,0]^d$. Denote $\R_+^d=[0,\infty)^d$.

The \emph{norm} of a set $M$ is defined as $\|M\|=\sup\{\|x\|:\; x\in
M\}$, where $\|x\|$ is the Euclidean norm of $x\in\R^d$. A set $M$ is
said to be \emph{upper}, if $x\in M$ and $x\leq y$ imply that $y\in
M$, where all inequalities between vectors are understood
coordinatewisely. Inclusions of sets are always understood in the
non-strict sense, i.e. $M\subset L$ allows for the equality $M=L$.

The $\eps$-envelope $M^\eps$ of a closed set $M$ is defined as the set
of all points $x$ such that the distance between $x$ and the nearest
point of $M$ is at most $\eps$. The \emph{Hausdorff distance}
$\dha(M_1,M_2)$ between two closed sets $M_1$ and $M_2$ in $\R^d$ is
the smallest $\eps\geq 0$ such that $M_1\subset M_2^\eps$ and
$M_2\subset M_1^\eps$.
The Hausdorff distance metrises the family of compact sets, while it
can be infinite for unbounded sets.

The \emph{support function} (see \cite[Sec.~1.7]{schn}) of a set $M$
in $\R^d$ is defined as
\begin{displaymath}
  h_M(u)=\sup\{\langle u,x\rangle:\; x\in M\}\,,\qquad u\in\R^d\,,
\end{displaymath}
where $\langle u,x\rangle$ denotes the scalar product.  The support
function may take infinite values if $M$ is not bounded. Denote by
\begin{displaymath}
  M'=\{u:\; |h_M(u)|\neq\infty\}
\end{displaymath}
the \emph{effective domain} of the support function of $M$. The set
$M'$ is always a convex cone in $\R^d$. If $K$ is a cone in $\R^d$,
then $K'$ equals the \emph{dual cone} to $K$ defined as
\begin{equation}
  \label{eq:k-dual}
  K'=\{u\in\R^d:\,\langle u,x\rangle\leq 0\,
  \textrm{ for all }x\in K\}\,.
\end{equation}

\subsection{Set-valued portfolios}
\label{sec:set-valu-portf}

Let $\X$ be an almost surely non-empty \emph{random closed convex set}
in $\R^d$ (shortly called random set) that represents all feasible
terminal gains on $d$ assets expressed in physical units. The random
set $\X$ is called \emph{set-valued portfolio}. Assume that
$\X$ is defined on a complete non-atomic probability space
$(\Omega,\salg,\P)$. Any attainable terminal gain is a random vector
$\xi$ that almost surely takes values from $\X$, i.e.\ $\xi\in\X$
a.s., and such $\xi$ is called a \emph{selection} of $\X$.
We refer to \cite{mo1} for the modern mathematical theory of random
sets.

Since the free disposal of assets is allowed, with each point $x$, the
set $\X$ also contains all points dominated by $x$ coordinatewisely
and so $\X$ is a \emph{lower} set in $\R^d$. The \emph{efficient part}
of $\X$ is the set $\partial^+\X$ of all points $x\in\X$ such that no
other point of $\X$ dominates $x$ in the coordinatewise order. While
$\X$ itself is never bounded, $\X$ is called \emph{quasi-bounded} if
$\partial^+\X$ is a.s. bounded.

Fix $p\in[1,\infty]$ and consider the space $L^p(\R^d)$ of
$p$-integrable random vectors in $\R^d$ defined on
$(\Omega,\salg,\P)$.  The reciprocal $q$ is defined from
$p^{-1}+q^{-1}=1$. The $L^p$-norm of $\xi$ is denoted by
$\|\xi\|_p$. Furthermore, the family of $p$-integrable selections of
$\X$ is denoted by $\sel^p(\X)$, and $\sel^\infty(\X)$ is the family of
all essentially bounded selections.

In the following we assume that $\X$ is \emph{$p$-integrable},
i.e. $\X$ possesses at least one $p$-integrable selection.  A random
closed set is called $p$-integrably bounded if its norm if
$p$-integrable (a.s. bounded if $p=\infty$). In the case of
set-valued portfolios, this property is considered for $\partial^+\X$.

The closed sum $\X+\Y$ is defined as the random set being the closure
of $\X(\omega)+\Y(\omega)$ for $\omega\in\Omega$. It is shown in
\cite[Th.~1.4]{hia:ume77} (see also \cite[Prop.~2.1.4]{mo1}) that the
set of $p$-integrable selections of $\X+\Y$ coincides with the norm
closure of $\sel^p(\X)+\sel^p(\Y)$ if $p\in[1,\infty)$. If $p=\infty$,
then $L^\infty(\X+\Y)$ is a norm closed set that contains the closed
sum $\sel^\infty(\X)+\sel^\infty(\Y)$.

\subsection{Examples of set-valued portfolios}
\label{sec:examples-set-valued}


\begin{example}[Univariate portfolios]
  \label{ex:univ-port}
  If $d=1$, then $\X=(-\infty,X]$ is a half-line and the monotonicity
  of risks implies that it suffices to consider only its upper bound
  $X$ as in the classical theory of risk measures.
\end{example}

\begin{example}[Exchange cones]
  \label{ex:cone}
  Let $X\in L^p(\R^d)$ represent gains from $d$ assets. Furthermore,
  let $\K\supset\R^d_-$ be a convex (distinct from the whole space and
  possibly random) \emph{exchange cone} representing the family of
  portfolios available at price zero. Its symmetric variant
  $\check{\K}$ is the solvency cone, while the dual cone $\K'$
  contains all consistent price systems, see
  \cite{kab:saf09,schach01}. Formally, $\K$ is a random closed set
  with values being cones.
  Define $\X=X+\K$, so that selections of $\X$ correspond to
  portfolios that are possible to obtain from $X$ following the
  exchange rules determined by $\K$. 
  
  If $\K$ does not contain any line (and is called a proper cone),
  then the market has an \emph{efficient friction}. Otherwise, some
  exchanges are free from transaction costs. In this case different
  random vectors $X$ yield the same portfolio $X+\K$, which is also an
  argument in favour of working directly with set-valued
  portfolios. The cone $\K$ is a half-space if and only if all
  exchanges do not involve transaction costs. In difference to
  \cite{bent:lep13}, our setting does not require that the exchange
  cone is proper.

  If $\K$ is deterministic, then we denote it by $K$. If $K=\R_-^d$,
  no exchanges are allowed.
\end{example}

\begin{example}[Cones generated by bid-ask matrix]
  \label{ex:bm-cone}
  In the case of $d$ currencies, the cone $\K$ is usually generated by a
  \emph{bid-ask matrix}, as in Kabanov's transaction costs model, see
  \cite{kab99,schach01}. Let $\Pi=(\pi^{(ij)})$ be a (possibly
  random) matrix of exchange rates, so that $\pi^{(ij)}$ is the number
  of units of currency $i$ needed to buy one unit of currency $j$. It
  is assumed that the elements of $\Pi$ are positive, the diagonal
  elements are all one and $\pi^{(ij)}\leq\pi^{(ik)}\pi^{(kj)}$
  meaning that a direct exchange is always cheaper than a chain of
  exchanges.  The cone $\K$ describes the family of portfolios
  available at price zero, so that $\K$ is spanned by vectors $-e_i$
  and $e_j-\pi^{(ij)}e_i$ for $i,j=1,\dots,d$, where $e_1,\dots,e_d$
  are standard basis vectors in $\R^d$.
  If the gain $X$ contains derivatives drawn on the exchange rates,
  then we arrive at the situation when $X$ and the exchange cone $\K$
  are dependent.
\end{example}

\begin{example}[Conical setting with constraints]
  \label{ex:eligible}
  It is possible to modify the conical setting by requiring that all
  positions acquired after trading are subject to some linear or other
  constraints, see e.g. \cite{far:koc:mun13}. This amounts to
  considering the intersection of $\X=X+\K$ with a linear subspace or
  a more general subset of $\R^d$, that (if a.s. non-empty) results in
  a possibly non-conical set-valued portfolio and provides another
  motivation for working with set-valued portfolios.
\end{example}

\cite{pen:pen10} study in depth not necessarily conical transaction
costs models in view of the no-arbitrage property, see also
\cite{kav:mol05}. One of the most important examples is the model of
currency markets with liquidity costs or exchange constraints.

The following examples describe several non-conical models that yield
quasi-bounded set-valued portfolios. Despite the fact that some of
them are generated by random vectors, it is essential to treat these
portfolios as random sets, e.g. for possible diversification
effects. The latter means that a sum of such set-valued portfolios is
not necessarily equal to the set-valued portfolio generated by the sum
of the generating random vectors.

\begin{example}[Restricted liquidity]
  \label{ex:illiquid}
  Let $\X=X+K$, where $K=\{x:\; \sum x_i\leq 0, x_i\leq
  1,i=1,\dots,d\}$. 
  Then the exchanges up to the unit volume are at the unit rate free
  from transaction costs while other exchanges are not allowed. A
  similar example with transaction costs and a random exchange cone
  $\K$ can be constructed as $\X=X+(\K\cap (a+\R_-^d))$ for some
  $a\in\R_+^d$. A more general variant from \cite[Ex.~2.4]{pen:pen10}
  models liquidity costs depending on the transaction's volume.
\end{example}

\begin{example}
  \label{ex:port-poly}
  Let $X^{(1)},\dots,X^{(n)}$ be random vectors in $\R^d$ that
  represent terminal gains in $d$ lines (e.g. currencies) of $n$
  investments.  The random set $\X$ is defined as the set of all points
  in $\R^d$ dominated by at least one convex combination of the
  gains. 
  In other words, $\X$ is the sum of $\R_-^d$ and the convex hull of
  $X^{(1)},\dots,X^{(n)}$. 

  If $X^{(1)}=(X_1,X_2)$ and $X^{(2)}=(X_2,X_1)$ for $d=2$ and a
  bivariate random vector $(X_1,X_2)$, then $\X$ describes an
  arbitrary profit allocation between two different lines without
  transaction costs up to the amount $|X_1-X_2|$.
\end{example}

\begin{example}
  \label{eq:ball}
  Assume that $\X=X+B_R+\R_-^d$, where $B_R$ is the ball of
  fixed radius $R$ centred at the origin. This model corresponds to
  the case, when infinitesimally small transactions are free to
  exchange at the rate that depends on the balance between the
  portfolio components.
\end{example}

\begin{example}[Transactions maintaining solvency]
  \label{ex:short-sales}
  Let $\K$ be an exchange cone from Example~\ref{ex:cone} and let $X$
  be the value of a portfolio. Define $\X$ to be the set of points
  coordinatewisely dominated by a point from $(X+\K)\cap\check{\K}$ if
  $X$ belongs to the solvency cone $\check{\K}=\{-x:\; x\in \K\}$, and
  $\X=X+\R_-^d$ if $X\notin\check{\K}$.
  In this case no transactions are allowed in
  the non-solvent case and otherwise all transactions should maintain
  the solvency of the portfolio. 
\end{example}

\subsection{Set-valued risk measure}
\label{sec:set-valued-risk}

The following definition is adapted from
\cite{ham:hey10,ham:hey:rud11} and \cite{kul08}, where it appears in
the exchange cones setting.

\begin{definition}
  \label{def:svrm}
  A function $\rhogen(\X)$ defined on $p$-integrable set-valued
  portfolios is called a \emph{set-valued coherent risk measure} if it
  takes values being upper convex sets and satisfies the following
  conditions.
  \begin{enumerate}
  \item $\rhogen(\X+a)=\rhogen(\X)-a$ for all $a\in\R^d$ (cash
    invariance).
  \item If $\X\subset\Y$ a.s., then $\rhogen(\X)\subset\rhogen(\Y)$
    (monotonicity).
  \item $\rhogen(c\X)=c\rhogen(\X)$ for all $c>0$ (homogeneity).
  \item $\rhogen(\X\bplus\Y)\supset \rhogen(\X)+\rhogen(\Y)$
    (subadditivity).
  \end{enumerate}
  The risk measure $\rhogen$ is said to be \emph{closed-valued} if its
  values are closed sets.  Furthermore, $\rhogen$ is said to be a
  \emph{convex} set-valued risk measure if the homogeneity and
  subadditivity conditions are replaced by
  \begin{displaymath}
    \rhogen(t\X\bplus(1-t)\Y)\supset t\rhogen(\X)+(1-t)\rhogen(\Y)\,.
  \end{displaymath}
\end{definition}

The names for the subadditivity and convexity properties are justified
by the fact that sets can be ordered by the reverse inclusion; we
follow \cite{ham:hey10,ham:hey:rud11} in this respect.
Definition~\ref{def:svrm} appears in \cite{kul08} and
\cite{ham:hey10,ham:hey:rud11}, with the argument of $\rhogen$ being a
random vector $X$ and for a fixed exchange cone $\K$, which in our
formulation means that the argument of $\rhogen$ is the random set
$X+\K$.

The set-valued portfolio $\X$ is \emph{acceptable} if
$0\in\rhogen(\X)$. The subadditivity of $\rhogen$ means that the
acceptability of $\X$ and $\Y$ entails the acceptability of
$\X\bplus\Y$, as in the classical case of coherent risk measures.
In the univariate case, $\X=(-\infty,X]$ and
$\rhogen(\X)=[\risk(X),\infty)$ for a coherent risk measure $\risk$,
so that $\X$ is acceptable if and only if $\risk(X)\leq 0$.  If
$0\in\X$ a.s., then $\X$ is acceptable under any closed-valued
coherent risk measure. Indeed, then
$\rhogen(\X)\supset\rhogen(\R_-^d)$, while $\rhogen(\R_-^d)$ contains
the origin by the homogeneity and subadditivity properties and the
closedness of the values for $\rhogen$. For a non-coherent
$\rhogen$, it is sensible to extra impose the normalisation condition
$\rhogen(K)=\check{K}$ for each deterministic exchange cone $K$.

\begin{examplec}[Capital requirements in the exchange cones setting]
  \label{rem:cap-req}
  The value of the risk measure $\rhogen(\X)$ determines 
  capital amounts $a\in\R^d$ that make $\X+a$ acceptable.  The
  necessary capital should be allocated at time zero, when the
  exchange rules are determined by a non-random exchange cone
  $K_0$. Thus, the initial capital $x$ should be chosen so that
  $x+K_0$ intersects $\rhogen(X+\K)$, and
  \begin{displaymath}
    A_0=\rhogen(X+\K)+\check{K}_0
  \end{displaymath}
  is the family of all possible initial capital requirements.  Optimal
  capital requirements are given by the extremal points from $A_0$ in
  the order generated by the cone $K_0$.  If $\K=K_0$ is not random,
  $A_0=\rhogen(X+K_0)$. If $K_0$ is a half-space, meaning that the
  initial exchanges are free from transaction costs, then $A_0$ is a
  half-space too. In this case, the sensible initial capital is given
  by the tangent point to $\rhogen(X+\K)$ in direction of the normal
  to $K_0$, see Example~\ref{ex:intro}.

  If $A_0$ is the whole space, which might be the case, for instance,
  if $\K$ and $K_0$ are two different half-spaces, then
  it is possible to release an infinite capital from the position, and
  this situation should be excluded for the modelling purposes.
\end{examplec}

\section{Selection risk measure for set-valued portfolios}
\label{sec:select-risk-meas}

\subsection{Acceptability of set-valued portfolios}
\label{sec:accept-set-valu}

Below we explicitly construct set-valued risk measures based on
selections of $\X$. Let $\risk_1,\dots,\risk_d$ be law invariant
coherent risk measures defined on the space $L^p(\R)$ with values in
$\R$. Furthermore, assume that each $\risk_i$ satisfies the Fatou
property, which for $p=\infty$ follows from the law invariance
\citep{jouin:sch:touz06} and for $p\in[1,\infty)$ is always the case
if $\risk_i$ takes only finite values, see
\cite[Th.~3.1]{kain:rues09}. For a random vector
$\xi=(\xi_1,\dots,\xi_d)\in L^p(\R^d)$ write
\begin{displaymath}
  \vecrisk(\xi)=(\risk_1(\xi_1),\dots,\risk_d(\xi_d))\,.
\end{displaymath}
Random vector $\xi$ is said to be \emph{acceptable} if
$\vecrisk(\xi)\leq 0$, i.e. $\risk_i(\xi_i)\leq0$ for all
$i=1,\dots,d$. This is exactly the case if portfolio $\X=\xi+\R_-^d$
is acceptable with respect to the set-valued measure 
\begin{displaymath}
  \rhogen(\X)=\vecrisk(\xi)+\R_+^d=\times_{i=1}^d [\risk_i(\xi_i),\infty)\,.
\end{displaymath}
It is a special case of the \emph{regulator risk measure} considered
in \citep{ham:rud:yan13}. The following definition suggests a possible
acceptability criterion for set-valued portfolios that leads to a risk
measure satisfying the axioms from Definition~\ref{def:svrm}.  

\begin{definition}
  A $p$-integrable set-valued portfolio $\X$ is said to be
  \emph{selection acceptable} (in the following simply called
  acceptable) if $\vecrisk(\xi)\leq0$ for at least one selection
  $\xi\in\sel^p(\X)$.
\end{definition}

The monotonicity property of univariate risk measures
$\risk_1,\dots,\risk_d$ implies that $\X$ is acceptable if and
only if its efficient part $\partial^+\X$ admits an acceptable
selection. 

\begin{examplec}
  The acceptability of $\X=X+\K$ means that it is possible to transfer
  the assets given by the components of $X$ according to the exchange
  rules determined by $\K$, so that the resulting random vector
  $X+\eta$ with $\eta\in\sel^p(\K)$ has all acceptable components.
\end{examplec}

\begin{remark}[Generalisations]
  \label{rem:general}
  The acceptability of selections can be judged using any other
  multivariate coherent risk measure, e.g. considered in
  \cite{ham:hey:rud11}, that are not necessarily of point plus cone
  type, or numerical multivariate risk measures from
  \cite{bur:rues06,ekel:gal:hen12} and
  \cite{far:koc:mun13}. Furthermore, it is possible to consider
  acceptability of a general convex subset of $L^p(\R^d)$ that might
  not be interpreted as the family of selections of a set-valued
  portfolio. This may be of advantage in the dynamic setting
  \citep{fein:rud14c} or
  for considering uncertainty models as in \citep{bion:ker12}.
\end{remark}

\subsection{Coherency of selection risk measure}
\label{sec:coher-sel-risk}

\begin{definition}
  \label{def:rhonot}
  The \emph{selection risk measure} of $\X$ is defined as
  the set of deterministic portfolios $x$ that make $\X+x$ acceptable,
  i.e.
  \begin{equation}
    \label{eq:rhonot}
    \rhonot(\X)=\{x\in\R^d:\,\X+x\;\text{is acceptable}\}\,.
  \end{equation}
  Its closed-valued variant is $\rhos(\X)=\cl\rhonot(\X)$.
\end{definition}

\begin{theorem}
  \label{thr:rhos}
  The selection risk measure $\rhonot$ defined by \eqref{eq:rhonot}
  and its closed-valued variant $\rhos$ are law invariant set-valued
  coherent risk measures, and
  \begin{equation}
    \label{eq:union}
    \rhonot(\X)=\bigcup_{\xi\in\sel^p(\X)} (\vecrisk(\xi)+\R_+^d)\,.
  \end{equation}
\end{theorem}
\begin{proof}
  We show first that $\rhonot(\X)$ is an upper set.  Let
  $x\in\rhonot(\X)$ and $y\geq x$. If $\xi+x$ is acceptable for some
  $\xi\in\sel^p(\X)$, then also $\xi+y$ is acceptable because of the
  monotonicity of the components of $\vecrisk$. Hence $y\in\rhos(\X)$.
  If $x\in\rhos(\X)$ and $y\geq x$, then $x_n\to x$ for a sequence
  $x_n\in\rhonot(\X)$. Consider any $y'>y$, then $x_n\leq y'$ for
  sufficiently large $n$, so that $y'\in\rhonot(\X)$. Letting $y'$
  decrease to $y$ yields that $y\in\rhos(\X)$, so that $\rhos(\X)$ is
  an upper set.

  In order to confirm the convexity of $\rhonot(\X)$, assume that
  $x,y\in\rhonot(\X)$ with $\vecrisk(\xi+x)\leq 0$ and
  $\vecrisk(\eta+y)\leq 0$ and take any $\lambda\in(0,1)$. The
  subadditivity of components of $\vecrisk$ implies that
  \begin{align*}
    \vecrisk(\lambda\xi+(1-\lambda)\eta+\lambda x+(1-\lambda)y)
    &=\vecrisk(\lambda(\xi+x)+(1-\lambda)(\eta+y))\\
    &\leq \lambda \vecrisk(\xi+x)+(1-\lambda)\vecrisk(\eta+y)\leq 0\,.
  \end{align*}
  It remains to note that $\lambda\xi+(1-\lambda)\eta$ is a selection
  of $\X$ in view of the convexity of $\X$. Then $\rhos(\X)$ is convex
  as the closure of a convex set.

  The law invariance property is not immediate, since identically
  distributed random closed sets might have rather different families
  of selections, see \cite[p.~32]{mo1}. Denote by $\salg_{\X}$ the
  $\sigma$-algebra generated by the random closed set $\X$, see
  \cite[Def.~1.2.4]{mo1}. If $\X$ is acceptable, then
  $\vecrisk(\xi)\leq0$ for some $\xi\in\sel^p(\X)$. The dilatation
  monotonicity of law invariant numerical coherent risk measures on a
  non-atomic probability space (see \cite{cher:grig07}) implies that
  \begin{displaymath}
    \vecrisk(\E(\xi|\salg_{\X}))\leq \vecrisk(\xi)\leq 0\,.
  \end{displaymath}
  Therefore, the conditional expectation $\eta=\E(\xi|\salg_{\X})$ is
  also acceptable. The convexity of $\X$ implies that $\eta$ is a
  $p$-integrable $\salg_{\X}$-measurable selection of $\X$. Therefore,
  $\X$ is acceptable if and only if it has an acceptable
  $\salg_{\X}$-measurable selection. It remains to note that two
  identically distributed random sets have the same families of
  selections which are measurable with respect to the minimal
  $\sigma$-algebras generated by these sets, see
  \cite[Prop.~1.2.18]{mo1}. In particular, the intersections of these
  families with $L^p(\R^d)$ are identical. Thus, $\rhonot$ and $\rhos$
  are law invariant.

  Representation \eqref{eq:union} follows from the fact that
  $\rhonot(\X)$ is the union of $\{x:\; \vecrisk(\xi+x)\leq
  0\}=\vecrisk(\xi)+\R_+^d$ for $\xi\in\sel^p(\X)$.

  The first two properties of coherent risk measures follow directly
  from the definition of $\rhonot$.  The homogeneity and subadditivity
  follow from the fact that all acceptable random sets build a
  cone. Indeed, if $\xi\in\sel^p(\X)$ is acceptable, then $c\xi$ is an
  acceptable selection of $c\X$ and so $c\X$ is acceptable. If
  $\xi\in\sel^p(\X)$ and $\eta\in\sel^p(\Y)$ are acceptable, then
  $\xi+\eta$ is acceptable because the components of $\vecrisk$ are
  coherent risk measures and so $\sel^p(\X\bplus\Y)$ contains an
  acceptable random vector. Thus, $\rhonot(\X\bplus\Y)\supset
  \rhonot(\X)+\rhonot(\Y)$ and by passing to the closure we arrive at
  the subadditivity property of $\rhos$.
\end{proof}

Representation \eqref{eq:union} was used in \cite{ham:rud:yan13} to
define the \emph{market extension} of a regulator risk measure in the
conical models setting.

\begin{remark}[Eligible portfolios]
  \label{ex:eligible}
  In order to simplify the presentation it is assumed throughout that
  all portfolios $a\in\R^d$ can be used to offset the risk.  Following
  the setting of \cite{ham:hey10,ham:hey:rud11,ham:rud:yan13}, it is
  possible to assume that the set of \emph{eligible portfolios} is a
  proper linear subspace $M$ of $\R^d$. The corresponding set-valued
  risk measure is $\rhos(\X)\cap M$, which equals the union of
  $(\vecrisk(\xi)+\R_+^d)\cap M$ for all $\xi\in\sel^p(\X)$.
\end{remark}

\subsection{Properties of selection risk measures}
\label{sec:prop-select-risk}

Conditions for closedness of set-valued risk measures in the exchange
cone setting were obtained in \cite{fein:rud14fs}.  The following result
establishes the closedness of $\rhonot(\X)$ for portfolios with
$p$-integrably bounded essential part.
Further results concerning closedness of $\rhonot$ are presented in
Corollary~\ref{cor:rhos-cases} and Corollary~\ref{cor:kul-rhos}.

\begin{theorem}
  \label{thm:closed}
  If $\partial^+\X$ is $p$-integrably bounded, then the selection risk
  measure $\rhonot(\X)$ is closed.
\end{theorem}
\begin{proof}
  Let $\xi_n+x_n$ be acceptable and $x_n\to x$. Note that
  $\sel^p(\partial^+\X)$ is weak compact for $p=1$ by
  \cite[Th.~2.1.19]{mo1} and for $p\in(1,\infty)$ by its boundedness
  in view of the reflexivity of $L^p(\R^d)$. By passing to
  subsequences we can assume that $\xi_n$ weakly converges to $\xi$ in
  $L^p(\R^d)$. The dual representation of coherent risk measures
  yields that for a random variable $\eta_n\in L^p(\R)$,
  $\risk_i(\eta_n)$ equals the supremum of $\E(\eta_n\zeta)$ over a
  family of random variables $\zeta$ in $L^q(\R)$. If $\eta_n$
  weakly converges to $\eta$ in $L^p(\R)$, then $\E(\eta_n
  \zeta)\to\E(\eta \zeta)$, so that
  $\risk_i(\eta)\leq\liminf\risk_i(\eta_n)$. Applying this argument to
  the components of $\xi_n$ we obtain that
  \begin{equation}
    \label{eq:in-r}
    \vecrisk(\xi+x)\leq \liminf \vecrisk(\xi_n+x_n)\leq 0\,,
  \end{equation}
  whence $x\in\rhonot(\partial^+\X)$.  If $p=\infty$, then
  $\{\xi_n\}$ are uniformly bounded and so have a subsequence that
  converges in distribution and so can be realised on the same
  probability space as an almost surely convergent sequence. The Fatou
  property of the components of $\vecrisk$ yields \eqref{eq:in-r}.
\end{proof}

\begin{theorem}[Lipschitz property]
  \label{thr:lip}
  Assume that all components of $\vecrisk$ take finite values on
  $L^p(\R)$.  Then there exists a constant $C>0$ such that
  $\dha(\rhos(\X),\rhos(\Y))\leq C \|\dha(\X,\Y)\|_p$ for all
  $p$-integrable set-valued portfolios $\X$ and $\Y$.
\end{theorem}
\begin{proof}
  For each $\xi\in\sel^p(\X)$ and $\eps>\|\dha(\X,\Y)\|_p$, there
  exists $\eta\in\sel^p(\Y)$ such that $\|\xi-\eta\|_p\leq \eps$. The
  Lipschitz property of $L^p$-risk measures (see
  \cite[Lemma~4.3]{foel:sch04} for $p=\infty$ and \cite{kain:rues09}
  for $p\in[1,\infty)$) implies that
  $\|\vecrisk(\xi)-\vecrisk(\eta)\|\leq C\eps$ for a constant $C$. By
  \eqref{eq:union}, the Hausdorff distance between $\rhos(\X)$ and
  $\rhos(\Y)$ is bounded by $C\eps$.
\end{proof}

\begin{example}[Selection expectation]
  \label{ex:expect}
  If $\vecrisk(\xi)=\E(-\xi)$ is the expectation of $-\xi$, then
  \begin{displaymath}
    \rhos(\X)=\E \check{\X}+\R_+^d\,,
  \end{displaymath}
  where $\E\check{\X}$ is the \emph{selection expectation} of
  $\check{\X}$, i.e. the closure of the set of expectations $\E\xi$
  for all integrable selections $\xi\in\sel^1(\X)$, see
  \cite[Sec.~2.1]{mo1}.  Thus, the selection risk measure yields a
  subadditive generalisation of the selection expectation.
\end{example}

\begin{example}
  \label{ex:essinf}
  Assume that $p=\infty$ and all components of $\vecrisk$ are given by
  $\risk_i(\xi_i)=-\,\mathrm{essinf}\, \xi_i$. Then $\X$ is acceptable
  if and only if $\X\cap\R_+^d$ is almost surely non-empty, and
  $\rhos(\X)$ is the set of points that belong to $\check{\X}$ with
  probability one.
\end{example}

\begin{example}
  \label{ex:only-cone}
  If $\X=\K$ is an exchange cone, then $\rhos(\K)$ is a deterministic
  convex cone that contains $\R_+^d$. If $\K=K$ is deterministic, then
  $\rhos(K)=\check{K}$. 
\end{example}

\begin{remark}
  \label{rem:prop}
  Selection risk measures have a number of further properties.

  \noindent
  \textbf{I}. Assume that $\vecrisk=(\risk,\dots,\risk)$ has all
  identical components.  If $\X$ is acceptable, then its orthogonal
  projection on the linear subspace $\HH$ of $\R^d$ generated by any
  $u_1,\dots,u_k\in\R_+^d$ is also acceptable, and so the risk of the
  projected $\X$ contains the projection of $\rhos(\X)$. 

  \noindent
  \textbf{II}. The conditional expectation of random set $\X$ with
  respect to a $\sigma$-algebra $\mathfrak{B}$ is defined as the
  closure of the set of conditional expectations for all its
  integrable selections, see \cite[Sec.~2.1.6]{mo1}. The dilatation
  monotonicity property of components of $\vecrisk$ implies that if
  $\xi$ is acceptable, then $\E(\xi|\mathfrak{B})$ is
  acceptable. Therefore, $\rhos$ is also \emph{dilatation monotone}
  meaning that
  \begin{displaymath}
    \rhos(\E(\X|\mathfrak{B}))\supset \rhos(\X)\,.
  \end{displaymath}
  In particular, $\rhos(\X)\subset \rhos(\E
  \X)=\E\check{\X}$. Therefore, the integrability of the support
  function $h_\X(u)$ for at least one $u$ provides an easy condition
  that guarantees that $\rhos(\X)$ is not equal to the whole space.
\end{remark}

\begin{remark}
  \label{ex:kul-joint}
  In the setting of Example~\ref{ex:cone}, $\rhos(X+\K)$ written as a
  function of $X$ only becomes a centrally symmetric variant of the
  coherent utility function considered in \cite[Def.~2.1]{kul08}. It
  should be noted that the utility function from \cite{kul08} and risk
  measures from \cite{ham:hey:rud11} depend on both $X$ and $\K$ and
  on the dependency structure between them and so are \emph{not law
    invariant} as function of $X$ only, if $\K$ is random.
\end{remark}

If the components of $\vecrisk$ are \emph{convex} risk measures, so that the
homogeneity assumption is dropped, then $\rhos(\X)$ is a convex
set-valued risk measure, which is not necessarily homogeneous. If the
components of $\vecrisk$ are not law invariant, then $\rhos$ is a
possibly not law invariant set-valued risk measure.  The following two
examples mention \emph{non-convex} risk measures, which are also
defined using selections.

\begin{example}
  \label{ex:var}
  Assume that the components of $\vecrisk$ are general cash invariant
  risk measures without imposing any convexity properties, e.g.
  values-at-risk at the level $\alpha$, bearing in mind that the
  resulting selection risk measure is no longer coherent and not
  necessarily law invariant. Then $\X$ is acceptable if and only if
  there exists a selection $\xi$ of $\X$ such that $\Prob{\xi_i\geq
    0}\geq \alpha$ for all $i$.
\end{example}

\begin{example}
  \label{ex:var-cone}
  Let $K_0$ be a deterministic exchange cone and fix some acceptance
  level $\alpha$. Call random vector $\xi$ acceptable if
  $\Prob{\xi\in\check{K}_0}\geq \alpha$ and note that this condition
  differs from requiring that $\Prob{\xi_i\geq 0}\geq \alpha$ for all
  $i$. Then a set-valued portfolio $\X$ is acceptable if and only if
  $\Prob{\X\cap\check{K}_0\neq\emptyset}\geq \alpha$.  If
  $\X=X+\R_-^d$ and $K_0=\R_-^d$, then $\{x:\;\Prob{X\geq
    -x}\geq\alpha\}$ is sometimes termed a multivariate quantile or
  the value-at-risk of $X$, see \cite{emb:puc06} and \cite{ham:hey10}.
\end{example}

\section{Bounds for selection risk measures}
\label{sec:bounds-set-valued}

\subsection{Upper bound}
\label{sec:upper-bound}

The family of selections of a random set is typically very rich. An
\emph{upper bound} for $\rhos(\X)$ can be obtained by restricting the
choice of possible selections. The convexity property of $\rhos(\X)$
implies that it contains the convex hull of the union of
$\vecrisk(\xi)+\R_+^d$ for the chosen selections $\xi$. This convex
hull corresponds to the case of a higher risk than $\rhos(\X)$. Making
the upper bound tighter by considering a larger family of selections
is in the interest of the agent in order to reduce the
capital requirements.

At first, it is possible to consider deterministic selections, also
called fixed points of $\X$, i.e. the points which belong to $\X$
with probability one. If $a\in \X$ a.s., then $\rhos(\X)\supset
-a+\R_+^d$.  However, this set of fixed points is typically rather
poor to reflect essential features related to the variability of $\X$.

Another possibility would be to consider selections of $\X$ of the
form $\xi+a$ for a fixed random vector $\xi$ and a deterministic
$a$. If $\X\supset \xi+M$ for a deterministic set $M$ (which always
can be chosen to be convex in view of the convexity of $\X$), then
\begin{equation}
  \label{eq:rhol-gen}
  \rhos(\X)\supset \vecrisk(\xi)+\R_+^d+\check{M}\,.
\end{equation}
It is possible to tighten the bound by taking the convex hull for the
union of the right-hand side for several $\xi$.  The inclusion in
\eqref{eq:rhol-gen} can be strict even if $\X=\xi+M$, since taking
\emph{random} selections of $M$ makes it possible to offset the risks
as the following example shows.

\begin{example}
  \label{ex:unit-ball}
  Let $\X=\xi+M$, where $M$ is the unit ball and $\xi=(\xi_1,\xi_2)$ is the
  standard bivariate normal vector. Consider the risk measure $\vecrisk$ with
  two identical components being expected shortfalls at level
  $0.05$. Then $\vecrisk(\xi)+M$ is the upper
  set generated by the ball of radius one centred at
  $\vecrisk(\xi)=(2.063,2.063)$. Consider the selection of $M$ given by
  $\eta=(\one_{X_1<X_2},\one_{X_1>X_2})$. By numerical calculation of
  the risks, it is easily seen that $\vecrisk(\xi+\eta)=(1.22,1.22)$,
  which does not belong to $\vecrisk(\xi)+M$.
\end{example}

\subsection{Lower bound}
\label{sec:lower-bound}

Below we describe a \emph{lower bound} for $\rhos(\X)$, which is a
superset of $\rhos(\X)$ and is also a set-valued coherent risk measure
itself.  For $x,y\in\R^d$, $xy$ (resp. $x/y$) denote the vectors
composed of pairwise products (resp. ratios) of the coordinates of $x$
and $y$. If $M$ is a set in $\R^d$, then $My=\{xy:\; x\in M\}$. By
agreement, let $\{0\}/0=\R$ and $0/0=-\infty$.

Let $\sZb\subset L^q(\R^d)$ be a non-empty family of non-negative
$q$-integrable random vectors $Z=(\zeta_1,\dots,\zeta_d)$ in $\R^d$.
Recall that $\E(\check{\X}Z)$ denotes the selection expectation of
$\check{\X}$ with the coordinates scaled according to the components
of $Z$, see Example~\ref{ex:expect}.  It exists, since $\X$ is assumed
to possess at least one $p$-integrable selection and so
$\sel^1(\check{\X}Z)\neq\emptyset$. Define the set-valued risk measure
\begin{equation}
  \label{eq:rhosel}
  \rhosel(\X)=\bigcap_{Z\in\sZb} \frac{\E(\check{\X}Z)}{\E Z}\,,
\end{equation}
which is similar to the classical dual representation of coherent risk
measures, see \cite{delb02} and also corresponds to the
$(Q,w)$-representation of risk measure in conical models from
\cite[Th.~4.2]{ham:hey:rud11} that is similar to
\eqref{eq:dual-rhosel} from the following theorem.

\begin{theorem}
  \label{thr:rhoz-rm}
  Assume that $\sZb$ is a non-empty family of non-negative
  $q$-integrable random vectors.
  The functional $\rhosel(X)$ is a closed-valued coherent risk
  measure, and
  \begin{equation}
    \label{eq:dual-rhosel}
    \rhosel(\X)=\bigcap_{Z\in\sZb,u\in\R_+^d}\{x:\;
    \E \langle x,uZ\rangle\geq -\E h_{\X}(uZ)\}\,.
  \end{equation}
\end{theorem}
\begin{proof}
  The closedness and convexity of $\rhosel(X)$ follow from the fact
  that it is intersection of half-spaces; it is an upper set since the
  normals to these half-spaces belong to $\R_+^d$. It is evident that
  $\rhosel(X)$ is monotonic, cash invariant and homogeneous. In order
  to check the subadditivity, note that
  \begin{align*}
    \rhosel(\X\bplus\Y) &=\bigcap_{Z\in\sZb}
    \frac{\E[(\check{\X}\bplus\check{\Y})Z]}{\E Z}
    \supset \bigcap_{Z\in\sZb} \frac{\E(\check{\X}Z)}{\E Z}
    +\frac{\E(\check{\Y}Z)}{\E Z}\\
    &\supset \bigcap_{Z\in\sZb} \frac{\E(\check{\X}Z)}{\E Z}
    +\bigcap_{Z\in\sZb} \frac{\E(\check{\Y}Z)}{\E Z}\,.
  \end{align*}
  Recall that the support function of the expectation of a set equals
  the expected value of the support function. Since
  $\E(\check{\X}Z)$ is an upper set,
  \begin{align*}
    \rhosel(\X)&=\bigcap_{Z\in\sZb}
    \frac{\E(\check{\X}Z)}{\E Z}\\
    &=\bigcap_{Z\in\sZb}\bigcap_{u\in\R^d_-}
    \{x:\,\E h_{\check{\X}Z}(u)\geq\langle x,u\E Z\rangle\}\\
    &=\bigcap_{Z\in\sZb}\bigcap_{u\in\R^d_-}
    \{x:\,\E h_{\X}(-Zu)\geq\E \langle x,uZ\rangle\}\,,
  \end{align*}
  and we arrive at \eqref{eq:dual-rhosel} by replacing $u$ with $-u$. 
\end{proof}

The components of $\vecrisk=(\risk_1,\dots,\risk_d)$ admit the dual
representations
\begin{equation}
  \label{eq:dual-ri}
  \risk_i(\xi_i)=\sup_{\zeta_i\in\sZ_i}
  \frac{\E(-\xi_i\zeta_i)}{\E \zeta_i}\,,
  \quad i=1,\dots,d\,,
\end{equation}
where, for each $i=1,\dots,d$, $\sZ_i$ is the dual cone to the family
of random variables $\xi\in L^p(\R)$ such that $\risk_i(\xi)\leq 0$,
i.e. $\sZ_i$ consists of non-negative $q$-integrable random variables
$\zeta$ such that $\E(\zeta\xi)\geq 0$ for all $\xi$ with
$\risk_i(\xi)\leq 0$, see \cite{delb02,foel:sch04}. Note that
$\sZ_1,\dots\sZ_d$ are the \emph{maximal} families that provide the
dual representation \eqref{eq:dual-ri}. Despite $\sZ_i$ contains
a.s. vanishing random variables, letting $0/0=-\infty$ ensures the
validity of \eqref{eq:dual-ri}.

\begin{theorem}
  \label{thr:domin-exp}
  Assume that the components of $\vecrisk=(\risk_1,\dots,\risk_d)$
  admit the dual representations \eqref{eq:dual-ri}. Then
  $\rhos(\X)\subset \rhosel(\X)$ for any family $\sZb$ of
  $q$-integrable random vectors $Z=(\zeta_1,\dots,\zeta_d)$ such
  that $\zeta_i\in\sZ_i$, $i=1,\dots,d$.
\end{theorem}
\begin{proof}
  In view of \eqref{eq:dual-ri},
  \begin{displaymath}
    [\risk_i(\xi_i),\infty)=\bigcap_{\zeta_i\in\sZ_i}
    [\frac{\E(-\xi_i\zeta_i)}{\E\zeta_i},\infty)\,,
  \end{displaymath}
  so that
  \begin{displaymath}
    \vecrisk(\xi)+\R_+^d\subset \bigcap_{Z\in\sZb}
    (\frac{\E(-\xi Z)}{\E Z}+\R_+^d)\,.
  \end{displaymath}
  By \eqref{eq:union},
  \begin{align*}
    \rhos(\X)&\subset \cl\bigcup_{\xi\in\sel^p(\X)} \bigcap_{Z\in\sZb}
    (\frac{\E(-\xi Z)}{\E Z}+\R_+^d)\\
    & \subset \cl\bigcap_{Z\in\sZb} \bigcup_{\xi\in\sel^p(\X)}
    (\frac{\E(-\xi Z)}{\E Z}+\R_+^d)
    \subset \bigcap_{Z\in\sZb} \frac{\E(\check{\X}Z)}{\E Z}\,.
  \end{align*}
  Note that for each $\xi\in\sel^p(\X)$ and $a\in\R_+^d$, the random
  vector $(-\xi+a)Z$ is an integrable selection of $\check{\X}Z$. The
  closure is omitted, since the selection expectation is already
  closed by definition.
\end{proof}

\begin{corollary}
  \label{cor:proper}
  The selection risk measure $\rhos(\X)$ is not equal to the whole
  space if $h_{\X}(Zu)$ is integrable for some $u\in\R_+^d$ and
  $Z=(\zeta_1,\dots,\zeta_d)$ with $\zeta_i\in\sZ_i$, $i=1,\dots,d$.
\end{corollary}

It should be noted that the acceptability of $\X$ under the risk
measure $\rhosel(\X)$ does not necessarily imply the existence of a
selection with all acceptable marginals.

\begin{remark}
  \label{rem:li-rosel}
  The risk measure $\rhosel$ is not law invariant in general. It is
  possible to construct a law invariant (and also tighter) lower bound
  for the selection risk measure by extending $\sZb$ to
  $\tilde{\sZb}$, so that, with each $Z$, the family $\tilde{\sZb}$
  contains all random vectors $\tilde{Z}$ that share the distribution
  with $Z$.
\end{remark}

\begin{proposition}
  \label{prop:ident-comp}
  Let all components of $\vecrisk=(\risk,\dots,\risk)$ be identical
  univariate risk measures whose dual representation
  \eqref{eq:dual-ri} involves the same family $\sZ=\sZ_i$ of
  a.s. non-negative random variables. Consider the family $\sZb_0$
  that consists of all $Z=(\zeta,\dots,\zeta)/\E\zeta$ for all
  $\zeta\in\sZ$. Then
  \begin{equation}
    \label{eq:rhoselnot-def}
    \rhoselnot(\X)=\bigcap_{u\in\R_+^d} \{x:\; \langle
    x,u\rangle\geq
    \risk(h_\X(u))\}\,,
  \end{equation}
  where $\risk(h_\X(u))=-\infty$ if $h_\X(u)=\infty$ with positive
  probability.
\end{proposition}
\begin{proof}
  By \eqref{eq:dual-rhosel},
  \begin{displaymath}
    \rhoselnot(\X)=\bigcap_{\zeta\in\sZ,u\in\R_+^d}\{x:\;
    \langle x,u\rangle\E\zeta\geq -\E (h_{\X}(u)\zeta)\}
    =\bigcap_{u\in\R_+^d}\{x:\;
    \langle x,u\rangle\geq \sup_{\zeta\in\sZ}\frac{\E(-h_{\X}(u)\zeta)}{\E\zeta}\}\,,
  \end{displaymath}
  so it remains to identify the supremum as the dual representation
  for $\risk(h_\X(u))$.
\end{proof}

The bounds for selection measures for set-valued portfolios determined
by exchange cones are considered in the subsequent sections. Below we
illustrate Proposition~\ref{prop:ident-comp} on two examples of
quasi-bounded portfolios.

\begin{example}
  Consider portfolio $\X$ with $\partial^+\X$ being the segment in the
  plane with end-points $X^{(1)}$ and $X^{(2)}$. Then
  \begin{displaymath}
    \rhoselnot(\X)=\bigcap_{u\in\R_+^2}
    \{x:\; \langle x,u\rangle\geq \risk(\max(\langle X^{(1)},u\rangle,
    \langle X^{(2)},u\rangle))\}\,.
  \end{displaymath}
\end{example}

\begin{example}
  Consider portfolio $\X$ from Example~\ref{eq:ball}. 
  If all components of $\vecrisk$ are identical, then
  \begin{displaymath}
    \rhoselnot(\X)=\bigcap_{u\in\R_+^d}\{x:\;
    \langle x,u\rangle\geq \risk(\langle X,u\rangle)-R\|u\|\}\,.
  \end{displaymath}
\end{example}

\section{Conical market models}
\label{sec:rand-exch-cones}

\subsection{Random exchange cones}
\label{sec:rand-exch}

Let $\X=X+\K$ for $X\in L^p(\R^d)$ and a (possibly random) exchange
cone $\K$, see Example~\ref{ex:cone}. Then
\begin{equation}
  \label{eq:bound-k-random}
  \vecrisk(X)+\rhos(\K)\subset\rhos(X+\K)\subset \rhos(\E(X|\K)+\K)\,,
\end{equation}
where the first inclusion relation is due to the subadditivity of
$\rhos((X+\R_-^d)+\K)$ and the second one follows from the dilatation
monotonicity of law invariant risk measures. If $X$ and $\K$ are
independent, the lower bound becomes $-\E X+\rhos(\K)$.

Since $\E h_\X(uZ)$ is infinite unless $uZ$ almost surely belongs to
the dual cone $\K'$ and $h_\X(uZ)=\langle X,uZ\rangle$ for $uZ\in\K'$,
the lower bound \eqref{eq:dual-rhosel} turns into
\begin{equation}
  \label{eq:xkr}
  \rhosel(X+\K)=\bigcap_{u\in\R_+^d,\, Z\in\sZb,\,
    uZ\in\K'\,\text{a.s.}}
  \{x:\; \E \langle x,uZ\rangle\geq -\E \langle X,uZ\rangle\}\,.
\end{equation}
The right-hand side of \eqref{eq:xkr} corresponds to the dual
representation for set-valued risk measures from \cite{kul08} and
\cite{ham:hey:rud11}, where it is written as function of $X$ only.

Consider now the frictionless case, where $\K$ is a random half-space,
so that $\K'=\{t\upsilon:\; t\geq0\}$ for a random direction
$\upsilon\in\R_+^d$. Then
\begin{displaymath}
  \rhosel(X+\K)=\bigcap_{\zeta\in\sZ}
  \{x:\; \E(\langle x,\upsilon\rangle \zeta)\geq
  -\E(\langle X,\upsilon\rangle \zeta)\}\,,
\end{displaymath}
where $\sZ$ is a family of non-negative $q$-integrable random
variables $\zeta$ such that $\upsilon_i\zeta\in\sZ_i$ for all $i=1,\dots,d$.

\begin{example}[Bivariate frictionless random exchanges]
  \label{ex:half-random}
  Consider two currencies exchangeable at random rate $\pi=\pi^{(21)}$
  without transaction costs, see Example~\ref{ex:bm-cone}, so that the
  exchange cone $\K$ is the half-plane with normal $(\pi,1)$. Assume
  that $\vecrisk=(\risk,\risk)$ for two identical $L^p$-risk
  measures. The selection risk measure of $\X=X+\K$ is the closure of
  the set of $(\risk(X_1+\eta),\risk(X_2-\eta\pi))$ for $\eta\in
  L^p(\R)$ and $\eta\pi\in L^p(\R^d)$, see
  Example~\ref{ex:two-currencies} for a numerical illustration.

  The value of $\rhos(\K)$ is useful to bound the selection risk
  measure of $\X=X+\K$, see \eqref{eq:bound-k-random}.  Assume that
  both $\pi$ and $\pi^{-1}$ are $p$-integrable. For the purpose of
  computation of $\rhos(\K)$ it suffices to consider selections
  of the form $(\eta,-\eta\pi)$, where $\eta,\eta\pi\in
  L^p(\R)$. Furthermore, it suffices to consider separately almost
  surely positive and almost surely negative $\eta$. Then $\rhos(\K)$
  is a cone in $\R^2$ bounded by the two half-lines with slopes
  \begin{displaymath}
    \gamma_1=
    \sup_{\eta,\eta\pi\in L^p(\R_+)}
    \frac{\risk(-\eta\pi)}{\risk(\eta)}\,,
    \qquad
    \gamma_2=
    \inf_{\eta,\eta\pi\in L^p(\R_+)} \frac{\risk(\eta\pi)}{\risk(-\eta)}\,.
  \end{displaymath}
  The canonical choices $\eta=1$ and $\eta=1/\pi$ yield that
  \begin{displaymath}
    \frac{1}{\risk(1/\pi)}
    =\max(-\risk(-\pi),\frac{1}{\risk(1/\pi)})\leq \gamma_1
    \leq \gamma_2\leq \min(\risk(\pi),\frac{-1}{\risk(-1/\pi)})
    =\risk(\pi)\,,
  \end{displaymath}
  where the maximal and minimal elements above are obtained after
  applying the Jensen inequality to the dual representation of a
  univariate risk measure (observe that $f(x)=1/x$ is convex on the
  positive half-line). Therefore, $\rhos(\K)$ contains the cone with
  slopes given by $(\risk(1/\pi))^{-1}$ and $\risk(\pi)$.

  A lower bound for $\rhos(\K)$ relies on a lower bound for
  $\gamma_2$ and an upper bound for $\gamma_1$. Using the dual
  representation of $\risk$ as $\risk(\xi)=\sup_{\zeta\in\sZ}
  \E(-\zeta\xi)/\E\zeta$ for a family $\sZ\subset L^q(\R_+)$ of random
  variables with positive expectation, for (positive) $\eta,\pi$ we
  have
  \begin{displaymath}
    \frac{\risk(\eta\pi)}{r(-\eta)}
    =\frac{\sup_{\zeta_1\in\sZ}\E(-\eta\pi\zeta_1/\E\zeta_1)}{\sup_{\zeta_2\in\sZ}\E(\eta\zeta_2/\E\zeta_2)}
    =-\inf_{\zeta_1,\zeta_2\in\sZ}\frac{\E(\eta\pi\zeta_1/\E\zeta_1)}{\E(\eta\zeta_2/\E\zeta_2)}
    \geq-\inf_{\zeta\in\sZ}\frac{\E(\eta\pi\zeta)}{\E(\eta\zeta)}\,,
  \end{displaymath}
  where the inequality follows from setting $\zeta_1=\zeta_2$. 
  Then
  \begin{displaymath}
    \gamma_2\geq\inf_{\eta,\eta\pi\in L^p(\R_+)}
    \left\{-\inf_{\zeta\in\sZ}\frac{\E(\eta\pi\zeta)}{\E(\eta\zeta)}\right\}
    \geq-\inf_{\zeta\in\sZ}\sup_{\eta,\eta\pi\in L^p(\R_+)}
    \frac{\E(\eta\pi\zeta)}{\E(\eta\zeta)}\,.
  \end{displaymath}
  
  If $\risk$ is the expected shortfall $\ES_\alpha$, so that $\sZ$
  consists of all random variables taking value $1$ with probability
  $\alpha$ and $0$ otherwise, and if $\pi$ has continuous distribution
  with a convex support, then
  \begin{align*}
    \gamma_2&\geq 
    -\inf_{\zeta\in\sZ}\sup_{\eta,\eta\pi\in L^p(\R_+)}
    \frac{\E(\pi\zeta)(\eta\zeta)}{\E(\eta\zeta)}\\
    &\geq-\inf_{\zeta\in\sZ}\sup_{\eta,\eta\pi\in L^p(\R_+)}
    \frac{\E(\pi\zeta\eta)}{\E(\eta)}
    =-\inf_{\zeta\in\sZ}\esssup\left(\pi\zeta\right)
    =\VaR_\alpha(\pi)\,,
  \end{align*}
  where $\VaR_\alpha$ denotes the value-at-risk. With a similar argument,
  \begin{align*}
    \gamma_1\leq -\sup_{\zeta\in\sZ}\inf_{\eta,\eta\pi\in L^p(\R_+)}
    \frac{\E(\eta\pi\zeta)}{\E(\eta\zeta)}
    \leq \VaR_{1-\alpha}(\pi)\,,
  \end{align*} 
  see Figure~\ref{fig:pure-half} for the corresponding lower and upper
  bounds.
  For any $\alpha\leq 1/2$ the upper bound for $\gamma_1$ is not
  greater than the lower bound for $\gamma_2$ meaning that $\rhos(\K)$
  is distinct from the whole space.
\end{example}

The arguments presented in Example~\ref{ex:half-random} and
\eqref{eq:bound-k-random} yield the following result.

\begin{proposition}
  \label{prop:half-plane}
  Let $d=2$. Consider the selection measure $\rhos$ generated by
  $\vecrisk$ with all components being expected shortfalls at level
  $\alpha\leq 1/2$. If $\X\subset X+\K$, where $\K$ is a half-plane
  with normal $(\pi,1)$ such that $\pi,\pi^{-1}\in L^p(\R_+)$, $\pi$
  has convex support, and $X$ and $\pi$ are independent, then
  $\rhos(\X)$ is distinct from the whole plane.
\end{proposition}

\subsection{Deterministic exchange cones}
\label{sec:determ-exch-cone}

Assume that $\X=X+K$ for a deterministic exchange cone $K$ and a
$p$-integrable random vector $X$. Since $\rhos(K)=\check{K}$,
\eqref{eq:bound-k-random} yields that
\begin{displaymath}
  \vecrisk(X)+\check{K}\subset\rhos(X+K)\subset -\E X+\check{K}\,.
\end{displaymath}

\begin{proposition}
  \label{prop:ident-comp-not}
  Assume that $\vecrisk$ has all identical components being $\risk$.
  Then
  \begin{equation}
    \label{eq:bound-detK}
    \rhoselnot(X+K)=\bigcap_{u\in K'}\{x:\,\risk(\langle X,u\rangle)
    \leq\langle x,u\rangle\}\,.
  \end{equation}
\end{proposition}
\begin{proof}
  The result follows from Proposition~\ref{prop:ident-comp} and the
  fact that $h_{X+K}(u)=\langle X,u\rangle$ if $u\in K'$ and
  otherwise the support function is infinite.
\end{proof}

A risk measure of $X+K$ for a deterministic cone $K$ is said to
\emph{marginalise} if its values are translates of $\check{K}$ for all
$X\in L^p(\R^d)$. The coherency of $\risk$ implies that the
intersection in \eqref{eq:bound-detK} can be taken over all $u$ being
extreme elements of $K'$, which in dimension 2, yields that
$\rhoselnot(X+K)$ is a translate of cone $\check{K}$. In general
dimension, if $K$ is a \emph{Riesz cone} (i.e. $\R^d$ with the order
generated by $K$ is a Riesz space), then $\rhoselnot(X+K)=a+\check{K}$
with $a\in\R^d$ being the supremum of $\E(-X\zeta)/\E\zeta$ in the
order generated by $K$.  Similar risk measures were proposed in
\cite[Ex.~6.6]{cas:mol07}, where instead of
$\{\E(-X\zeta)/\E\zeta:\,\zeta\in\sZ\}$ a depth-trimmed region was
considered.

\begin{example}
  \label{ex:tce}
  Let all components of $\vecrisk$ be univariate expected
  shortfalls $\mathrm{ES}_\alpha$ at level $\alpha$, so that $\sZ$ is
  the family of indicator random variables $\zeta=\one_A$ for all
  measurable $A\subset\Omega$ with $\P(A)>\alpha$ for some fixed
  $\alpha\in(0,1)$. Then $\rhoselnot(X+K)$ becomes the vector-valued
  worst conditional expectation (WCE) of $X$ introduced in
  \cite[Ex.~2.5]{jouin:med:touz04}. Furthermore, \eqref{eq:bound-detK}
  yields the set-valued expected shortfall (ES) defined as
  \begin{displaymath}
    \mathrm{ES}_\alpha(X+K)
    =\bigcap_{u\in K'}\{x:\,\mathrm{ES}_\alpha(\langle X,u\rangle)
    \leq\langle x,u\rangle\}\,.
  \end{displaymath}
  Its explicit expression if $K$ is a Riesz cone is given in
  \cite[eq.~(7.1)]{cas:mol07} --- in that case
  $\mathrm{ES}_\alpha(X+K)$ is a translate of $\check{K}$. Since the
  univariate ES and WCE coincide for non-atomic random variables,
  their multivariate versions coincide by
  Proposition~\ref{prop:ident-comp} if $X$ has a non-atomic
  distribution. It should be noted that $\mathrm{ES}_\alpha(X+K)$ is
  only a lower bound for the selection risk measure defined by
  applying $\mathrm{ES}_\alpha$ to the individual components of
  selections of $X+K$. In particular, the acceptability of $X+K$ under
  $\mathrm{ES}_\alpha(X+K)$ does not necessarily imply the existence
  of an acceptable selection of $X+K$.
\end{example}

The following result deals with the \emph{comonotonic} case, see
\cite{foel:sch04} for the definition of comonotonicity and
comonotically additive risk measures.

\begin{theorem}
  \label{thr:comonotonic-detK}
  Assume that $\vecrisk$ has all identical comonotonic additive
  components $\risk$.  If the components of $X$ are comonotonic and
  $K$ is a deterministic exchange cone, then
  $\rhos(X+K)=\vecrisk(X)+\check{K}$.
\end{theorem}
\begin{proof}
  Since $\risk$ is comonotonic additive and $X$ is comonotonic,
  $\risk(\langle X,a\rangle)=\langle \vecrisk(X),a\rangle$ for any
  $a\in K'\subset\R^d_+$.  Then
  \begin{displaymath}
    \vecrisk(X)+\check{K}=\bigcap_{a\in K'}\{x:\,
    \langle \vecrisk(X),a\rangle\leq\langle x,a\rangle\}\,,
  \end{displaymath}
  and consequently $\rhos(X+K)=\vecrisk(X)+\check{K}=\rhoselnot(X+K)$.
\end{proof}

Notice that $\vecrisk(X)$ is solely determined by the marginal
distributions of $X$. Consequently, if $\tilde{X}$ is a comonotonic
rearrangement of $X$ (i.e. a random vector with the same marginal
distributions and comonotonic coordinates), then
\begin{displaymath}
  \rhos(X+K)\supset\vecrisk(X)+\check{K}=\vecrisk(\tilde{X})+\check{K}
  =\rhos(\tilde{X}+K)\,.
\end{displaymath}

The following result shows that $\rhos(X+K)$ does not change if
instead of all selections of $X+K$ one uses only those being
deterministic functions of $X$.

\begin{theorem}
  \label{thr:dil-mon}
  Set-valued portfolio $\X=X+K$ is acceptable if and only if
  $\vecrisk(X+\eta(X))\leq 0$ for a selection $\eta(X)\in K$, which is
  a deterministic function of $X$.
\end{theorem}
\begin{proof}
  If $\vecrisk(X+\eta(X))\leq 0$, then $X$ is acceptable. For the
  reverse implication, the dilatation monotonicity of the components of
  $\vecrisk$ yields that
  \begin{displaymath}
    \vecrisk(X+\E(\eta|X))\leq \vecrisk(X+\eta)\leq 0\,.
  \end{displaymath}
  It remains to note that the conditional expectation $\E(\eta|X)$ is
  a function of $X$ taking values from $K$, since $K$ is a
  deterministic convex cone.
\end{proof}

\begin{example}[Frictionless market with deterministic exchange rates]
  Consider a vector $X\in L^\infty(\R^d)$ of terminal gains on $d$ currencies that can
  be exchanged without transaction costs, so that $K$ is a half-space
  with normal $u=(1,\pi^{(21)},\dots,\pi^{(d1)})$, where $\pi^{(j1)}$
  is the number of units of currency $j$ needed to by one unit of
  currency one, see Example~\ref{ex:bm-cone}.
  
  Let $\vecrisk=(\risk_1,\dots,\risk_d)$ with possibly different
  components whose maximal dual representations involve families of
  $q$-integrable random variables $\sZ_1,\dots,\sZ_d$ from
  \eqref{eq:dual-ri}. It will be shown later on in
  Corollary~\ref{cor:kul-rhos} that
  $\rhosel(X+K)=\rhos(X+K)$ for some family $\sZb$ satisfying the
  condition of Theorem~\ref{thr:domin-exp}.

  Then $Z=(\zeta_1,\dots,\zeta_d)\in K'=\{\lambda u:\lambda\geq 0\}$
  if and only if $\zeta_i=u_i\zeta$, $i=1,\dots,d$, for a random
  variable $\zeta$ from $\sZ_*=\cap_i \sZ_i$. The family $\sZ_*$
  determines the coherent risk measure $\risk_*$ called the convex
  convolution of $\risk_1,\dots,\risk_d$, see \cite{delb12}. Then
  $X$ is acceptable under $\rhos$ if and only if $\risk_*(\langle
  X,u\rangle)\leq 0$ and
  \begin{displaymath}
    \rhos(X+K)=\rhosel(X+K)
    =\{x:\; \langle x,u\rangle \geq \risk_*(\langle X,u\rangle)\}\,.
  \end{displaymath}
  If all components of $X$ represent the same currency, but the
  regulator requirements applied to each component are different,
  e.g. because they represent gains in different countries, then
  $u=(1,\dots,1)$ and $X$ is acceptable if and only if
  $\risk_*(X_1+\cdots+X_d)\leq 0$.
\end{example}

While the above examples and Theorem~\ref{thr:comonotonic-detK}
provide $\rhos(X+K)$ of the marginalised form $\vecrisk(X)+\check{K}$,
this is not always the case, as Example~\ref{ex:non-margin} confirms.

The calculation of $\rhos$ for $d=2$ and $p=\infty$ can be facilitated
by using the following result. It shows that $\rhos(X+K)$ coincides
with its upper bound $\rhoselnot(X+K)$ sufficiently far away from the
origin.

\begin{proposition}
  \label{prop:meetk}
  Assume that $p=\infty$ and $\vecrisk$ has all identical components
  $\risk$.  If $\X=X+K$ for a deterministic cone $K$ and an
  essentially bounded random vector $X$ in dimension $d=2$, then for
  all $x$ with sufficiently large norm, $x\in\rhoselnot(X+K)$ implies
  that $x\in\rhos(X+K)$.
\end{proposition}
\begin{proof}
  In dimension 2, we can always assume that the cone $K$ is generated
  by a bid-ask matrix (see Example~\ref{ex:bm-cone}), so that $K$ and
  its dual are given by
  \begin{displaymath}
    K=\{\lambda b_1+\delta b_2:\,\lambda,\delta\geq 0\}\,,\quad
    \quad K'=\{\lambda a_1+\delta a_2:\,\lambda,\delta\geq 0\}\,,
  \end{displaymath}
  where
  \begin{align*}
    b_1=(1,-\pi^{(21)})\,,\quad\quad & b_2=(-\pi^{(12)},1)\\
    a_1=(\pi^{(21)},1)\,,\quad\quad & a_2=(1,\pi^{(12)})\,.
  \end{align*}
  By Proposition~\ref{prop:ident-comp-not} and by the coherence of
  $\risk$,
  \begin{equation}\label{eq:rhoselnotbidim}
    \rhoselnot(X+K)
    =\bigcap_{i=1,2}\{x:\,\risk(\langle X,a_i\rangle)
    \leq\langle x,a_i\rangle\}\,.
  \end{equation}
  In order to obtain a point lying on the boundaries of both
  $\rhos(X+K)$ and $\rhoselnot(X+K)$, consider the positive random
  variable
  \begin{displaymath}
    \zeta=(X_1-\essinf X_1)/\pi^{(12)}\,.
  \end{displaymath}
  Then
  \begin{displaymath}
    X+\zeta b_2=\left(\essinf X_1,X_2+(X_1-\essinf
      X_1)/\pi^{(12)}\right)
  \end{displaymath}
  has the first a.s. deterministic coordinate. Since $\zeta$ is
  a.s. non-negative, $X+\zeta b_2$ is a selection of $X+K$.  Define
  \begin{displaymath}
    x_1=\vecrisk(X+\zeta b_2)
    =\left(-\essinf X_1,\frac{\risk(\langle X,a_2\rangle)
        +\essinf X_1}{\pi^{(12)}}\right)\,.
  \end{displaymath}
  The fact that $X+\zeta b_2\in X+K$ guarantees that
  $x_1\in\rhos(X+K)\subset\rhoselnot(X+K)$.  Since $a_2\in\R^2_+$ and
  the first component of $X+\zeta b_2$ is constant,
  \begin{equation}
    \label{eq:b2}
    \langle x_1,a_2\rangle
    =\risk(\langle X+\zeta b_2,a_2\rangle)
    =\risk(\langle X,a_2\rangle)\,,
  \end{equation}
  where the last equality holds because $a_2$ and $b_2$ are
  orthogonal.  Because of (\ref{eq:rhoselnotbidim}) and (\ref{eq:b2}),
  $x_1$ lies on the supporting line of $\rhoselnot(X+K)$ which is
  normal to $a_2$, whence it lies on the boundaries of $\rhos(X+K)$
  and $\rhoselnot(X+K)$.

  By a similar argument,
  \begin{displaymath}
    x_2=\left(\frac{\risk(\langle X,a_1\rangle)
        +\essinf X_2}{\pi^{(21)}},-\essinf X_2\right)
  \end{displaymath}
  lies on the supporting line of $\rhoselnot(X+K)$ which is normal to
  $a_1$, so on the boundaries of $\rhos(X+K)$ and $\rhoselnot(X+K)$.

  Let $b$ be the radius of a closed ball centred at the origin and
  containing the triangle with vertices at $x_1$, $x_2$, and the vertex
  of cone $\rhoselnot(X+K)$. If $x\in\rhoselnot(X+K)$ with $\|x\|\geq
  b$, then clearly $x\in \lambda x_1+(1-\lambda)x_2+\check{K}$ for
  some $0\leq\lambda\leq 1$, which by the convexity of $\rhos(X+K)$
  and the fact that $\rhos(X+K)=\rhos(X+K)+\check{K}$ guarantees that
  $x\in\rhos(X+K)$.
\end{proof}

\section{Numerical risk measures for set-valued portfolios}
\label{sec:numer-risk-meas}

The set-valued risk measure $\rhos$ gives rise to several \emph{numerical}
coherent risk measures, i.e. functionals $\numrisk(\X)$ with values in
$\R\cup\{\infty\}$ that satisfy the following properties.
\begin{enumerate}
\item There exists $u\in\R^d$ such that $\numrisk(\X+a)=\numrisk(\X)-\langle
  a,u\rangle$ for all $a\in\R^d$.
\item If $\X\subset\Y$, then $\numrisk(\X)\geq \numrisk(\Y)$.
\item $\numrisk(c\X)=c\numrisk(\X)$ for all $c>0$.
\item $\numrisk(\X+\Y)\leq \numrisk(\X)+\numrisk(\Y)$.
\end{enumerate}

The canonical scalarisation construction relies on the support
function of $\rhos(\X)$.  Namely, 
\begin{displaymath}
  \numrisk_u(\X)=-h_{\rhos(\X)}(-u)=\inf\{\langle u,x\rangle:\;
  x\in \rhos(\X)\}
\end{displaymath}
is a law invariant coherent risk measure for each
$u\in\R_+^d$. Furthermore, $\X$ is acceptable under $\rhos$, i.e.
$0\in\rhos(\X)$, if and only if $\X$ is acceptable under $\numrisk_u$
for all $u\in\R_+^d$.
If the exchange cone is trivial, i.e. $\X=X+\R_-^d$, then
$\numrisk_u(\X)=\langle \vecrisk(X),u\rangle$ is given by a linear
combination of the marginal risks.

Let $Z\in L^q(\R^d)$. The \emph{max-correlation} risk measure of
$p$-integrable random vector $X$ is defined as
\begin{displaymath}
  \maxcor_Z(X)=\sup_{\tilde{Z}\sim Z} \frac{\E\langle
    -X,\tilde{Z}\rangle}{\E Z}\,,
\end{displaymath}
where the supremum is taken over all random vectors $\tilde{Z}$
distributed as $Z$, see \cite{bur:rues06,rues06}. A general coherent
numerical risk measure of $X$ can be represented as the supremum of
$\maxcor_Z(X)$ over a family $Z\in\sZb\subset L^q(\R^d)$. Then the set
of $x\in\R^d$ that make $X$ acceptable is given by
\begin{displaymath}
  \{x\in\R^d:\; \E\langle x,Z\rangle\geq
  -\E\langle X,\tilde{Z}\rangle,\, \tilde{Z}\sim Z,\, Z\in\sZb\}\,,
\end{displaymath}
which is $\rhosel$ given by \eqref{eq:xkr}.

In the spirit of \citep{far:koc:mun13}, it is possible to define
the numerical risk of $\X$ as infimum of the payoffs $f(\eta)$ for
eligible portfolios $\eta$, such that $\X+\eta$ is acceptable,
i.e. contains a selection with all acceptable components. In the
conical market model and selection risk measures setting, this can be
rephrased as infimum of $f(\eta)$ for $\eta$ from
$\sA+\sel^p(\check{\X})$, where $\sA$ is the subset of $L^p(\R^d)$
that consists of random vectors with all individually acceptable
components.

\section{Dual representation}
\label{sec:cont-dual-repr}

The representation of the selection risk measure by \eqref{eq:union}
can be regarded as its primal representation. In
this section we arrive at the dual representations of the selection
risk measure and also general set-valued coherent risk measures of
set-valued portfolios.

\subsection{Lipschitz space of random sets}
\label{sec:lipsch-space-rand}

In general, the support function $h_\X(u)$ at a deterministic $u$ may
be infinite with a positive probability. In order to handle this
situation, we identify a random set-valued portfolio with its support
function evaluated at random directions and consider some special
families of set-valued portfolios.

Fix a (possibly random) closed convex cone $\gc\subset\R_+^d$, define
$\gc_1=\{x\in\gc:\;\|x\|=1\}$. Since $\gc_1$ is compact, all its
selections are bounded. Note that $\sel^\infty(\gc_1)$ is endowed with
the $L^\infty$-metric. Let $\Lip^p(\gc)$ be the family of all
functions $h:\sel^\infty(\gc_1)\mapsto L^p(\R)$ such that $h$ is
uniformly $p$-integrable on $\sel^\infty(\gc_1)$ and is Lipschitz with
$p$-integrable Lipschitz constant $\lipc{h}$. The \emph{Lipschitz
  space} is $\Lip^p(\gc)$ with the norm $\|h\|_L$ defined as the
maximum of the $L^p$-norm of $\lipc{h}$ and the $L^p$-norm of the
supremum of $|h(Y)|$ over all $Y\in\sel^\infty(\gc_1)$.

A set-valued portfolio $\X$ is said to belong to the space
$\Lip^p(\gc)$ if the effective domain of its support function
$h_\X(\cdot)$ is $\gc$, and $h_\X$ belongs to the space
$\Lip^p(\gc)$. Since 
\begin{displaymath}
  h_{\X\bplus\Y}(Y)=h_{\X}(Y)+h_{\Y}(Y)\,,
\end{displaymath}
$h_\X$ provides a linear embedding of such set-valued portfolios into
$\Lip^p(\gc)$.

\begin{lemma}
  \label{lemma:lip-space}
  Assume that $\X=\K+\Y$, where $\K$ is a convex cone and $\Y$ is a
  random closed set with the $p$-integrable norm
  $\|\Y\|=\sup\{\|y\|:\; y\in \Y\}$. Then $\X$ belongs to
  $\Lip^p(\gc)$, where $\gc=\K'$ is the dual cone to $\K$ and $\|h_\X\|_L$
  is at most the $L^p$-norm of $\|\Y\|$.
\end{lemma}
\begin{proof}
  The domain of the support function of $\X$ is $\gc$, and
  $h_\X(u)=h_\Y(u)$ for all $u\in\gc$. It is known that the support
  function of a compact set is Lipschitz with the Lipschitz constant
  being the norm of this set, see \citep[Th.~F.1]{mo1}. Thus, both the
  Lipschitz constant of $h_\X$ and the supremum of $|h_\X(u)|$ for
  $u\in\gc_1$ are bounded by $\|\Y\|$. 
\end{proof}

A quasi-bounded portfolio $\X$ belongs to $\Lip^p(\R_+^d)$ if
$\partial^+\X$ is $p$-integrably bounded. If
$M=\{(x_1,x_2)\in\R_-^2:\; x_1x_2\geq 1\}$, then $\X=X+M$ does not
belong to $\Lip^p(\R_+^2)$, no matter what $X$ is.

\begin{examplec}
  The random set $\X=X+\K$ for an exchange cone $\K$ with an efficient
  friction belongs to $\Lip^p(\K')$ if $X$ is $p$-integrable, and
  $\|h_\X\|_L$ is bounded by the $L^p$-norm of $\|X\|$. In the case of some
  frictionless exchanges the choice of $X$ is not unique, and the above
  statement holds with $X$ chosen from the intersection of $\X$ and
  the linear hull of $\K'$. 

  Considering portfolios of the form $X+\K$ from the space
  $\Lip^p(\K')$ means that all of them share the same exchange cone
  $\K$. This is reasonable, since the diversification effects affect
  only the portfolio components, while the (possibly random) exchange
  cone remains the same for all portfolios. This setting also
  corresponds to one adapted in \cite{ham:hey:rud11,kul08} by assuming
  that the conical models share the same exchange cone.
\end{examplec}

Consider linear functionals acting on $\Lip^p(\gc)$ as
\begin{equation}
  \label{eq:l-func}
  \langle \mu,h\rangle=\E \int h(u)\mu(du)
  =\E \sum_{i=1}^n \eta_i h(Y_i)\,,
\end{equation}
where $\mu$ is a random signed measure with $q$-integrable weights
$\eta_1,\dots,\eta_n$ assigned to its atoms
$Y_1,\dots,Y_n\in\sel^\infty(\gc_1)$ and any $n\geq 1$.  Note that
\eqref{eq:l-func} does not change if $\mu$ is a signed measure
attaching weights $1$ or $-1$ to $Z_i=\eta_iY_i\in\sel^q(\gc)$,
$i=1,\dots,n$.
These functionals form a linear space and build a complete family, in
particular, the values of $\langle \mu,h_\X\rangle$ uniquely identify
the distribution of $\X$. Indeed, by the Cram\'er--Wold device, it
suffices to take $\mu$ with atoms located at selections of $\gc$
scaled by $q$-integrable random variables in order to determine the
joint distribution of the values of $h_\X$ at these selections. 

\subsection{Weak-star convergence and dual representation}
\label{sec:weak-star-conv}

It is known \citep[Th.~4.3]{joh70} that a sequence of functions from a
Lipschitz space with values in a Banach space $E$ weak-star converges
if and only if the norms of the functions are uniformly bounded and
their values at each given argument weak-star converge in $E$. Note
that the weak convergence in Lipschitz spaces has not yet been
characterised, see \cite{han97}. In our case $E=L^p(\R)$. Thus, the
sequence $\X_n$ weak-star converges to $\X$ if and only if \\
i) $\|h_{\X_n}\|_L$ are uniformly bounded in $n$;\\
ii) for each $Y\in\sel^\infty(\gc_1)$, $h_{\X_n}(Y)$ converges to
$h_\X(Y)$ weakly in $L^p(\R)$ if $p\in[1,\infty)$ and in probability
if $p=\infty$.

\begin{lemma}
  \label{lemma:w-star}
  Let $\X$ and $\X_n$, $n\geq1$, belong to the Lipschitz space
  $\Lip^\infty(\gc)$.  If $\X_n$ weak-star converges to $\X$ in
  $\Lip^\infty(\gc)$, then the Hausdorff distance $\dha(\X_n,\X)$
  converges to zero in probability and $\X_n\subset M+\gc'$ a.s. for
  a deterministic compact set $M$ and all $n$.
\end{lemma}
\begin{proof}
  The weak-star convergence yields the uniform boundedness for the
  $L^\infty$-norms of $h_{\X_n}$, so that $|h_{\X_n}(u)|$ is bounded
  by a constant $c$ for all $u\in\gc_1$, meaning that $\X_n$ is a subset
  of $M+\gc'$ for a deterministic set $M$.  Furthermore, the Lipschitz
  constant of $h_{\X_n}$ is bounded by $c$. Let $Y_1,\dots,Y_k$ be
  a (random) $\eps$-net in $\gc_1$. Then
  \begin{align*}
    \dha(\X_n,\X)=\sup_{u\in\gc_1}|h_{\X_n}(u)-h_\X(u)|
    \leq \sum_{i=1}^k |h_{\X_n}(Y_i)-h_\X(Y_i)|
    +\eps c\,.
  \end{align*}
  It suffices to note that $h_{\X_n}(Y_i)$ converges in probability to
  $h_\X(Y_i)$ for each $i$.
\end{proof}

Consider a general set-valued coherent risk measure $\rhogen$ from
Definition~\ref{def:svrm}.  If the family 
\begin{displaymath}
  \sA_\rhogen=\{\X\in\Lip^p(\gc):\; \rhogen(\X)\ni 0\}
\end{displaymath}
of acceptable set-valued portfolios is weak-star closed, the risk
measure $\rhogen$ is said to satisfy the \emph{Fatou property}. In
view of the cash invariance property, this formulation of the Fatou
property is equivalent to $\rhogen(\X)\supset\limsup \rhogen(\X_n)$ if
$\X_n$ weak-star converges to $\X$ in $\Lip^p(\gc)$, cf. \cite{kul08}.
Recall that the \emph{upper limit} of $\rhogen(\X_n)$ is the set of
limits for all convergent subsequences of $\{x_n\}$, where $x_n\in
\rhogen(\X_n)$, $n\geq1$, see \cite[Def.~B.4]{mo1}.

\begin{theorem}[Dual representation]
  \label{thr:dual-sets}
  A function $\rhogen$ on random sets from $\Lip^p(\gc)$ with values
  being convex closed upper sets is a set-valued coherent risk measure
  with the Fatou property if and only if $\rhogen(\X)=\rhosel(\X)$
  given by \eqref{eq:dual-rhosel} for a certain family
  $\sZb\subset\sel^q(\gc)$.
\end{theorem}
\begin{proof}
  \textsl{Necessity.}  Note first that \eqref{eq:dual-rhosel} can be
  equivalently written as
  \begin{equation}
    \label{eq:dual-rep-mu}
    \rhosel(\X)=\{x:\; \langle x,\E \int u\mu(du)\rangle
    \geq -\E \int h_\X(u)\mu(du)\,,\; \mu\in\sM\}
  \end{equation}
  for $\sM$ being the family of counting measures with atoms
  from $\sZb$.

  The dual cone $\sM$ to $\sA_\rhogen$ is the family of signed
  measures $\mu$ on $\gc$ with $q$-integrable total variation such
  that $\langle \mu,h_\X\rangle \geq 0$ for all
  $\X\in\sA_\rhogen$. Since $\X=\Y+\gc'$ is acceptable for each
  random compact convex set $\Y$ containing the origin, each
  $\mu\in\sM$ is non-negative.  By the bipolar theorem, $\sA_\rhogen$
  is the dual cone to $\sM$.

  \textsl{Sufficiency.}  By Theorem~\ref{thr:rhoz-rm}, $\rhosel$ is a
  set-valued coherent risk measure. By an application of the dominated
  convergence theorem and noticing that the weak-star convergence
  implies the uniform boundedness of the norms, its acceptance set is
  weak-star closed.
\end{proof}

\begin{examplec}
  If $\X=X+\K$ for a random exchange cone $\K$ and $p$-integrable
  random vector $X$, then $\X\in\Lip^p(\gc)$ with $\gc=\K'$ and
  $h_\X(Y)=\langle X,Y\rangle$ for all $Y\in L^\infty(\gc_1)$. If
  $\X_n=X_n+\K$, then the weak-star convergence of $X_n$ to $X$ implies
  that $h_{\X_n}(Y)$ weakly converges to $h_\X(Y)$ in $L^p$ for
  $p\in[1,\infty)$ and in probability if $p=\infty$ for all
  $Y\in\Lip^p(\gc_1)$. Furthermore, the weak convergence of $X_n$
  implies the boundedness of their $L^p$-norms and so the uniform
  boundedness of $\|h_{\X_n}\|_L$. In case $p=\infty$, the uniform
  boundedness of the norms follows from the weak-star convergence of
  $X_n$ to $X$. Thus, $\X_n$ weak-star converges to $\X$.  In the case
  of (partially) frictionless models the random vectors $X_n$ and $X$
  should be chosen appropriately from the intersections of the
  corresponding set-valued portfolios with the linear hull of $\K'$.
  By Theorem~\ref{thr:dual-sets}, set-valued coherent risk measures
  defined on portfolios $\X=X+\K$ and satisfying the Fatou property
  can be represented by \eqref{eq:dual-rhosel},
  which is exactly the dual representation \cite{kul08} for $p=\infty$
  and from \cite{ham:hey:rud11} for a general $p$. 
\end{examplec}

\subsection{Fatou property of selection risk measure}
\label{sec:fatou-prop-select}

Since the selection risk measure $\rhos$ is a special case of a
general set-valued coherent risk measure, $\rhos=\rhosel$ for a
suitable (and possibly non-unique) family $\sZb$ provided $\rhos$
satisfies the Fatou property. Note that the Fatou property of $\rhos$
is weaker than the Fatou property of $\rhonot$, which is established
in the following theorems. Even for the simplest selection risk
measure being the selection expectation (see Example~\ref{ex:expect}),
the Fatou property is a rather delicate result proved in
\cite{bal:hes95}.

\begin{theorem}
  \label{thr:fp-rhos}
  Assume that $\xi\in\gc'$ a.s. and $\vecrisk(\xi)\leq 0$ imply that
  $\xi=0$ a.s.  If $p\in(1,\infty)$, then the selection risk measure
  $\rhonot$ satisfies the Fatou property on random sets from
  $\Lip^p(\gc)$.
\end{theorem}
\begin{proof}
  By the cash invariance property, it suffices to assume that $x_n\to
  0$ for $x_n\in\rhonot(\X_n)$, with the aim to show that
  $0\in\rhonot(\X)$.  Without loss of generality assume that
  $x_n=\vecrisk(\xi_n)$ for $\xi_n\in L^p(\X_n)$, $n\geq1$. 

  Assume first that $\sup_n \|\xi_n\|_p<\infty$. Then there is a
  subsequence $\xi_{n_k}$ that weakly converges to $\xi$ in $L^p$.
  For each $Y\in L^\infty(\gc_1)$ and $Z\in L^q(\R_+)$,
  \begin{equation}
    \label{eq:xi-n-bound}
    \langle \xi_n,Y\rangle Z \leq h_{\X_n}(Y)Z. 
  \end{equation}
  In view of the weak-star convergence of $\X_n$ to $\X$, and passing
  to the limits as $n\to\infty$, we have
  \begin{displaymath}
    \E \langle \xi,Y\rangle Z \leq \E h_{\X}(Y)Z. 
  \end{displaymath}
  Thus, $\langle \xi,Y\rangle \leq h_{\X}(Y)$ a.s. meaning that $\xi$
  is a selection of $\X$. 

  The $L^p$-weak lower semicontinuity of the risk
  measure (see \cite[Th.~3.1]{kain:rues09}) yields that
  \begin{displaymath}
    \vecrisk(\xi)\leq\liminf \vecrisk(\xi_n)=0,
  \end{displaymath}
  where the lower limit is understood coordinatewisely, meaning that
  $\X$ is acceptable.

  Now assume that $\|\xi_n\|_p\to \infty$. The normalised sequence
  $\xi'_n=\xi_n/\|\xi_n\|_p$ is $L^p$-bounded and so admits a weakly
  convergent subsequence and so we assume without loss of generality
  that $\xi'_n$ weakly converges in $L^p$ to $\xi'$. Then $\xi'$ is
  acceptable since
  \begin{displaymath}
    \vecrisk(\xi')\leq\liminf \vecrisk(\xi'_n)=0. 
  \end{displaymath}
  By \eqref{eq:xi-n-bound}, $\E \langle \xi',Y\rangle Z\leq 0$ for all
  $Y\in L^\infty(\gc_1)$ and $Z\in L^q(\R_+)$.
  Thus, $\langle \xi',Y\rangle\leq 0$ a.s., so that $\xi'\in
  L^p(\gc')$ and $\vecrisk(\xi')\leq0$. This contradicts the imposed
  condition, since $\|\xi'\|_p=1$ and so $\xi'$ is not zero with a positive
  probability.
\end{proof}

The condition of Theorem~\ref{thr:fp-rhos} holds if the interior of
$\gc$ almost surely contains a deterministic point $u$ and $\vecrisk$
has all identical components $\risk$ such that $\risk(\eta)>0$ for any
non-trivial non-positive $\eta$. Indeed, then
$\eta=\langle\xi,u\rangle$ is strictly negative for
$\xi\in\gc'\setminus\{0\}$ contrary to
$\risk(\langle\xi,u\rangle)\leq\langle\vecrisk(\xi),u\rangle\leq 0$.

\begin{theorem}
  \label{thr:fp}
  For $p=\infty$, the selection risk measure $\rhonot$
  satisfies the \emph{Fatou property} on quasi-bounded set-valued
  portfolios.
\end{theorem}
\begin{proof}
  If necessary by passing to a subsequence, assume that
  $x_n\in\rhonot(\X_n)$ and $x_n\to 0$. Then, for all $n$ there exists
  $\xi_n\in\sel^\infty(\partial^+\X_n)$ such that
  $\vecrisk(\xi_n+x_n)\leq 0$.

  The weak-star convergence of $\X_n$ in
  $\Lip^\infty(\R_+^d)$ implies that $|h_{\X_n}(u)|\leq c$ for a
  constant $c$ and all $u\in\R_-^d$, $n\geq1$. Thus, $\partial^+\X_n$,
  $n\geq1$, are all subsets of a fixed compact set $M$, and the
  Hausdorff distance $\dha(\X_n,\X)$ converges to zero
  in probability as $n\to\infty$.
  By \cite[Th.~1.6.21]{mo1}, $\X_n$ converges to $\X$ in probability
  as random closed sets. The convergence in probability implies the
  weak convergence of random closed sets, see \cite[Cor.~1.6.22]{mo1},
  so that $\partial^+\X_n$ weakly converges to $\partial^+\X$ because
  the map $\partial^+$ is continuous on quasi-bounded
  portfolios. Since the family of convex subsets of a compact set is
  compact in the Hausdorff metric and $\xi_n\in M$, it is possible to
  find a subsequence of $(\partial^+\X_n,\xi_n)$ that weakly converges
  to $(\partial^+\X,\xi)$. Pass to the chosen subsequence and realise
  the pairs $(\partial^+\X_n,\xi_n)$ on the same probability space, so
  that $\partial^+\X_n\to\partial^+\X$ in the Hausdorff metric and
  $\xi_n\to\xi$ a.s. with $\xi\in\partial^+\X$ a.s.

  Since the components of $\vecrisk$ satisfy the Fatou property and
  $\{\xi_n\}$ are uniformly bounded,
  \begin{displaymath}
    \vecrisk(\xi)\leq \liminf\vecrisk(\xi_n+x_n)\leq 0\,.
  \end{displaymath}
  Thus, $0\in\rhonot(\X)$.
\end{proof}

Recall that $\sZ_1,\dots,\sZ_d$ denote the subsets of $L^q(\R_+)$ that
yield the maximal dual representations of the components of
$\vecrisk=(\risk_1,\dots,\risk_d)$, see \eqref{eq:dual-ri}. 

\begin{corollary}
  \label{cor:rhos-cases}
  The selection risk measure $\rhos(\X)$ on set-valued portfolios $\X$
  from $\Lip^p(\gc)$ for $p\in(1,\infty)$ (under conditions of
  Theorem~\ref{thr:fp-rhos}) and on quasi-bounded
  portfolios for $p=\infty$ equals $\rhonot(\X)$ and admits
  representation as $\rhosel(\X)$ for a family $\sZb\subset
  L^q(\R^d)$ such that each $Z=(\zeta_1,\dots,\zeta_d)\in\sZb$
  satisfies $\zeta_i\in\sZ_i$ for all $i=1,\dots,d$.
\end{corollary}
\begin{proof}
  Consider the sequence $\X_n=\X$, $n\geq1$. Then the Fatou property
  of $\rhonot$ established in Theorems~\ref{thr:fp-rhos} and
  \ref{thr:fp} implies that the upper limit of $\rhonot(\X)$ (being
  the closure of $\rhonot(\X)$) is a subset of $\rhonot(\X)$, so that
  $\rhonot(\X)$ is closed.

  Consider any $\xi=(\xi_1,\dots,\xi_d)\in L^p(\R^d)$ with
  $\vecrisk(\xi)\leq 0$. Then $\xi+\gc'$ is an acceptable set-valued
  portfolio, so that $0\in\rhosel(\xi+\gc')$. By
  \eqref{eq:dual-rhosel}, this is the case if and only if
  $\E\langle\xi,Z\rangle\geq 0$ for all $Z\in\sZb$. In turn, this is
  equivalent to $\E(\xi_i\zeta_i)\geq 0$ for all $\xi_i$ such that
  $\risk_i(\xi_i)\leq 0$ and $\zeta_i$ being the $i$th component of
  $Z$. Since $\sZ_i$ provides the maximal dual representation of
  $\risk_i$, it necessarily contains $\zeta_i$.
\end{proof}

The following result establishes the Fatou property for the selection
risk measure for $p=\infty$ and portfolios obtained as $\X=X+K$
for a deterministic exchange cone $K$.

\begin{theorem}
  \label{thr:fatou-det-K}
  For $p=\infty$, the selection risk measure $\rhonot$ satisfies the
  Fatou property on the family of portfolios obtained as $X+K$ for an
  essentially bounded random vector $X$ and a fixed deterministic
  exchange cone $K$.
\end{theorem}
\begin{proof}
  The weak-star convergence of $\X_n=X_n+K$ implies that the Hausdorff
  distance between $\X_n$ and $\X$ converges to zero in probability and $-c\leq
  h_{\X_n}(u)\leq c$ for a constant $c$ and all $u\in K'$. Since there
  exists $a\in K'$ such that $\langle a,u\rangle \geq \eps>0$ for all
  $u\in K'$, we have that $x+K\subset \X_n\subset M+K$ for some
  $x\in\R^d$ and a compact set $M$. If $K$ does not contain any line
  and $\X_n=X_n+K$, this means that $X_n$ is uniformly bounded and
  converges to $X$ in probability, where $\X=X+K$. If $K$ contains a
  line, then $X_n$ and $X$ can be chosen to satisfy this requirement,
  say by letting $X_n$ be the intersection of the largest affine space
  contained in $X_n+K$ with the linear hull of $K'$.

  For each $u\in K'$, the
  numerical risk measure $-h_{\rhonot(X+K)}(-u)$ defined as a function
  of a random vector $X$ is law invariant and coherent,
  so that it satisfies the Fatou property by
  \cite[Th.~2.6]{ekel:sch11} which also applies if the cash invariance
  property is relaxed as in \cite{bur:rues06}. Thus, if $X_n$
  converges to $X$ in probability and $\|X_n\|\leq 1$ a.s. for all
  $n$, then
  \begin{displaymath}
    h_{\rhonot(X+K)}(u)\geq \limsup h_{\rhonot(X_n+K)}(u)\,.
  \end{displaymath}
  If $x\in\limsup \rhonot(X_n+K)$, then $x=\lim x_{n_k}$ for a certain
  sequence $x_{n_k}\in\rhonot(X_{n_k}+K)$. Therefore,
  \begin{align*}
    \langle x,u\rangle =\lim \langle x_{n_k},u\rangle
    \leq \lim h_{\rhonot(X_{n_k}+K)}(u)\leq \limsup h_{\rhonot(X_n+K)}(u)\,.
  \end{align*}
  Thus,
  \begin{displaymath}
    h_{\rhonot(X+K)}(u)\geq h_{\limsup \rhonot(X_n+K)}(u)\,,\quad u\in K'\,,
  \end{displaymath}
  so that $\rhonot(X+K)$ contains the upper limit of $\rhonot(X_n+K)$.
\end{proof}

\begin{corollary}
  \label{cor:kul-rhos}
  If $p=\infty$, the selection risk measure $\rhonot(\X)$ for
  set-valued portfolios $\X=X+K$ with a fixed deterministic exchange
  cone $K$ has closed values and is equal to $\rhosel$ for a certain
  family of integrable random vectors $\sZb$ such that each
  $Z=(\zeta_1,\dots,\zeta_d)\in\sZb$ satisfies $\zeta_i\in\sZ_i$ for all
  $i=1,\dots,d$.
\end{corollary}
\begin{proof}
  It is shown in \cite{ham:hey:rud11} and \cite{kul08} that a
  set-valued risk measure with $p=\infty$ satisfying the Fatou
  property has the dual representation as $\rhosel$, and so $\rhos$
  admits exactly the same dual representation. By repeating the
  argument from the proof of Corollary~\ref{cor:rhos-cases}, it is
  seen that $\zeta_i\in\sZ_i$.
\end{proof}

It should be noted that not all risk measures $\rhosel$ are selection
risk measures, and so the acceptability of $\X$ under $\rhosel$ does
not immediately imply the existence of an acceptable selection and so
does not guarantee the existence of a trading strategy that eliminates
the risk.  While the calculation of the risk measure $\rhosel$ may be
rather complicated, its primal representation \eqref{eq:union} as the
selection risk measure opens a possibility for an approximation of its
values from above by exploring selections of $X+K$.

\section{Computation and approximation of risk measures}
\label{sec:comp-appr-risk}

The evaluation of $\rhos(\X)$ involves calculation of $\vecrisk(\xi)$
for all selections $\xi\in\sel^p(\X)$. The family of such selections
is immense, and in application only several possible selections can be
considered. A wider choice of selections is in the interest of the
agent in order to better approximate $\rhos(\X)$ from
above and so reduce the required capital reserves.  For lower bounds,
one can use the risk measure $\rhosel(\X)$ or its superset obtained by
restricting the family $\sZb$, e.g. $\rhoselnot(\X)$ if $\vecrisk$ has
all identical components.

In view of \eqref{eq:union} and the convexity of its values,
$\rhos(\X)$ contains the convex hull of the union of
$\vecrisk(\xi)+\R_+^d$ for any collection of selections
$\xi\in\sel^p(\X)$. It should be noted that this convex hull is not
subadditive in general (unless the family of selections builds a cone)
and so itself cannot be used as a risk measure, while providing a
reasonable approximation for it. It is possible to start with some
``natural'' selections $\xi'$ and $\xi''$ and consider all their
convex combinations or combine them as $\xi' 1_A+\xi'' 1_{A^c}$ for
events $A$.

For the deterministic exchange cone model $\X=X+K$, it is sensible to
consider selections $X+t\eta$ for all $t\geq 0$ and a selection $\eta$
taking values from the boundary of $K$. By Theorem~\ref{thr:dil-mon},
it suffices to work with selections of $K$ which are functions of $X$,
and write them as $\eta(X)$. Since the aim is to minimise the risk, it
is natural to choose $\eta$ which is a sort of ``countermonotonic''
with respect to $X$, while choosing comonotonic $X$ and $\eta$ does
not yield any gain in risk for their sum.

Assume that the components of $\vecrisk$ are expected shortfalls at
level $\alpha$.  Then in order to approximate $\rhos$, it is possible
to use a ``favourable'' selection $\eta_*\in K$ constructed by
projecting $X$ onto $K$ following the two-step procedure. First, $X$
is translated by subtracting the vector of univariate
$\alpha$-quantiles in order to obtain random vector $Y$ whose
univariate $\alpha$-quantiles are zero. Then $Y$ is projected onto the
boundary of the solvency cone $\check{K}$ and $\eta_*(X)$ is defined
as the opposite of such projection. If $Y$ belongs to the centrally
symmetric cone to $K'$, it is mapped to the origin and no compensation
will be applied in this case. Consider all selections of the form
$X+t\eta_*(X)$ for $t>0$. With this choice, the components of $X$
assuming small values are partially compensated by the remaining
components.

An alternative procedure is to modify the projection rule, so that
only some part of the boundary of the solvency cone is used for
compensation. In dimension $d=2$, $\eta_{1}$ and $\eta_{2}$ are
defined by projecting $Y$ onto either one of the two half-lines that
form the boundary of $\check{K}$.

\begin{example}
  \label{ex:non-margin}
  Consider the random vector $X$ taking values $(-2,4)$ and $(4,-2)$
  with equal probabilities. Let $K$ be the cone with the points
  $(-5,1)$ and $(1,-5)$ on its boundary, so the corresponding bid-ask
  matrix has entries $\pi^{(12)}=\pi^{(21)}=5$, see
  Example~\ref{ex:bm-cone}.  Assume that $\vecrisk$ consists of two
  identical components being the expected shortfall ${\rm ES}_\alpha$
  at level $\alpha=3/4$. Observe that ${\rm ES}_\alpha(X)=(0,0)$ and
  for
  \begin{displaymath}
    \eta_1=
    \begin{cases}
      (1.2,-6)&\textrm{ if }X=(-2,4)\,, \\
      (0,0)&\textrm{ if }X=(4,-2)\,,
    \end{cases}
    \quad\quad
    \eta_2=
    \begin{cases}
      (0,0)&\textrm{ if }X=(-2,4)\,, \\
      (-6,1.2)&\textrm{ if }X=(4,-2)\,,
    \end{cases}
  \end{displaymath}
  we obtain ${\rm ES}_\alpha(X+\eta_1)=(-0.8,2)$ and ${\rm
    ES}_\alpha(X+\eta_2)=(2,-0.8)$.

  By Theorem~\ref{thr:dil-mon}, the boundary of $\rhos(X+K)$ is given
  by $\vecrisk(X+\eta(X))$ with $\eta(X)$ belonging to the boundary of
  $K$. In order to compensate the risks of $X$ it is natural to choose
  $\eta$ as function of $X$ such that $\eta(-2,4)=t(-5,1)$ or
  $\eta(-2,4)=t(1,-5)$ and $\eta(4,-2)=s(-5,1)$ or
  $\eta(-2,4)=s(1,-5)$ for $t,s>0$. The minimisation problem over $t$
  and $s$ can be easily solved analytically (or numerically) and
  yields the boundary of $\rhos(X+K)$.  In Figure~\ref{fig:srisk}, the
  boundary of $\rhoselnot(X+K)$ is shown as dashed line, the boundary
  of $\vecrisk(X)+\check{K}$ is dotted line, while $\rhos(X+K)$ is the
  shaded region.

  \begin{figure}[h!]
    \begin{center}
      \includegraphics[scale=.45,angle=-90]{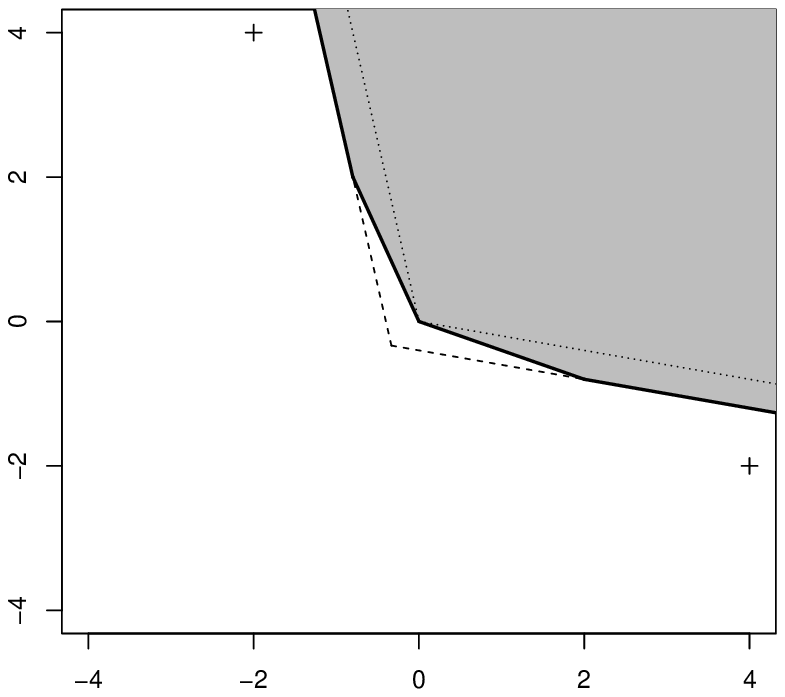}
      \caption{Two values for $X$, selection risk measure (shaded
        region) and its bounds from below and from above for
        Example~\ref{ex:non-margin}.\label{fig:srisk}}
    \end{center}
  \end{figure}
\end{example}

\begin{example}
  \label{ex:norm}
  Consider the cone $K$ with the points $(-1.5,1)$ and $(1,-1.5)$ on
  its boundary, so the corresponding bid-ask matrix has entries
  $\pi^{(12)}=\pi^{(21)}=1.5$. Let $\vecrisk=({\rm ES}_{0.05},{\rm
    ES}_{0.05})$ consist of two identical components. Let $X^{(n)}$
  follow the empirical distribution for a sample of $n$ observations
  from the bivariate normal distribution with i.i.d. components of
  mean $0.5$ and variance $1$.

  In order to approximate $\rhos(X^{(n)}+K)$ from above
  (i.e. construct a subset of $\rhos(X^{(n)}+K)$), we first determine
  the set $A$ consisting of $\vecrisk(\xi)$ for the selections $\xi$
  of $X+K$ obtained as $X^{(n)}+t\eta_*(X^{(n)})$,
  $X^{(n)}+t\eta_{1}(X^{(n)})$, and $X^{(n)}+t\eta_{2}(X^{(n)})$, for
  $t>0$ and $\eta_*$, $\eta_{1}$, $\eta_{2}$ described above using
  projection on $\check{K}$ and the two half-lines from its
  boundary. Then we determine the convex hull of $A$ and the points
  $x_1$ and $x_2$ described in the proof to
  Proposition~\ref{prop:meetk} and finally add $\check{K}$ to this
  convex hull.

  Figure~\ref{fig:srisknorm}(a) shows a sample of $n=1000$
  observations of a standard bivariate normal distribution and the
  approximation to the true value of $\rhos(X^{(n)}+K)$ described above,
  on the right panel a detail of the same plot is presented. The
  constructed approximation to $\rhos(X^{(n)}+K)$ is the shaded region,
  the boundary of $\rhoselnot(X^{(n)}+K)$ obtained as described in
  Proposition~\ref{prop:ident-comp-not} is plotted as dashed line, and
  the boundary of $\vecrisk(X^{(n)})+\check{K}$ is plotted as dotted line.

  \begin{figure}[h!]
    \begin{center}
      \begin{tabular}{cc}
        \subfigure[]{\includegraphics[scale=.45,angle=-90]{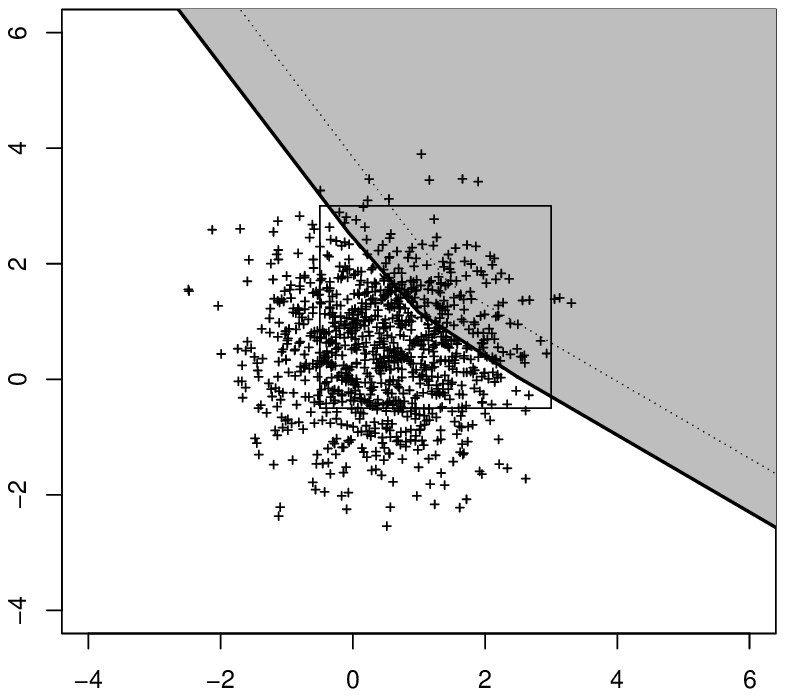}} &
        \subfigure[]{\includegraphics[scale=.45,angle=-90]{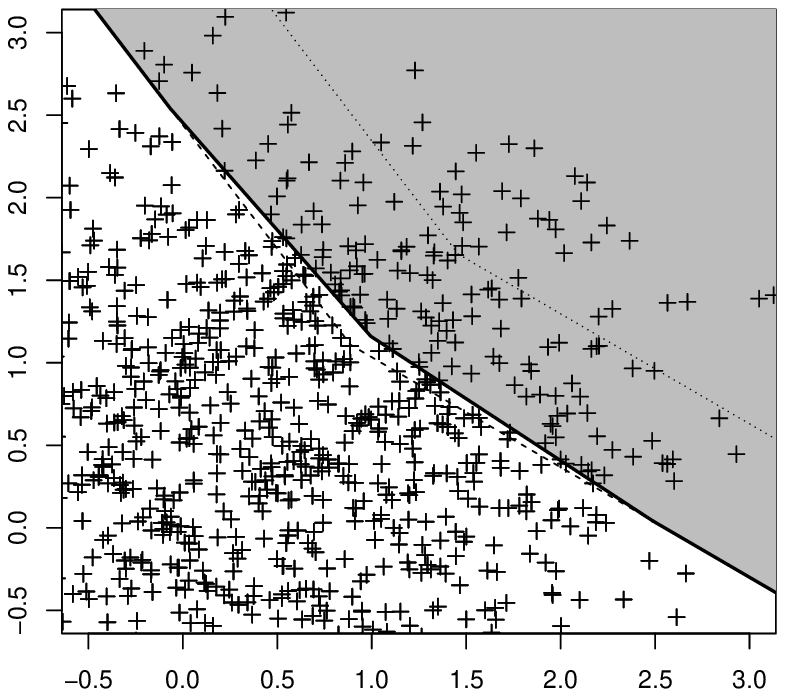}}
      \end{tabular}
      \caption{(a) An approximation to the selection risk measure and
        its bounds for a sample of normally distributed gains and a
        deterministic exchange cone; (b) enlarged part of the
        plot.\label{fig:srisknorm}}
    \end{center}
  \end{figure}
\end{example}

\begin{example}
  \label{ex:two-currencies}
  Consider the random frictionless exchange of two currencies
  described in Examples~\ref{ex:intro}
  and~\ref{ex:half-random}. Transactions that diminish the risk of
  $X=(X_1,X_2)$ can be constructed by projecting $X$ onto the boundary
  line to half-plane $\K$ with normal $(\pi,1)$ and then subtracting
  the scaled projection from $X$. This leads to the family of
  selections $X+t\eta_*$ for $t\geq 0$ and $\eta_*\in\K$ given by
  \begin{equation}
    \label{eq:ex_frictionless}
    \eta_*=-\left(\frac{X_1-\pi X_2}{1+\pi^2},\frac{\pi^2 X_2-\pi
      X_1}{1+\pi^2}\right)\,.
  \end{equation}
  Assume that the exchange rate $\pi=\pi^{(21)}$ is log-normally
  distributed with mean $\pi_0$ being the initial exchange rate and
  volatility $\sigma$. We approximate the risk of $X+\K$ with $X$
  having a bivariate standard normal distribution and independent of
  $\pi$ (equivalently independent of $\K$) for $\vecrisk$ with two
  identical components being ${\rm ES}_{\alpha}$. Observe that ${\rm
    ES}_{\alpha}(\pi)=-\alpha^{-1}\Phi(\Phi^{-1}(\alpha)-\sigma)\pi_0$
  and ${\rm ES}_{\alpha}(1/\pi)=-\alpha^{-1}
  e^{\sigma^2}(1-\Phi(\Phi^{-1}(1-\alpha)+\sigma))\pi_0^{-1}$, where
  $\Phi$ is the cumulative distribution function of a standard normal
  random variable and $\Phi^{-1}$ its quantile function.

  Fix $\alpha=0.05$, $\pi_0=1.5$ and the volatility $\sigma=0.4$. The
  bounds on $\rhos(\K)$ are given by two cones $C_1$ and $C_2$ whose
  boundaries are determined in Example~\ref{ex:half-random}, see
  Figure~\ref{fig:pure-half}. Namely $C_1$ is bounded by half-lines
  with slopes ${\rm ES}_{0.05}(\pi)=-0.613$ and $1/{\rm
    ES}_{0.05}(1/\pi)=-3.123$, while the half-lines determining $C_2$
  have slopes ${\rm VaR}_{0.05}(\pi)=-0.717$ and ${\rm
    VaR}_{0.95}(\pi)=-2.674$, see Example~\ref{ex:half-random}. 

  \begin{figure}
    \centering
    \includegraphics[scale=.45,angle=-90]{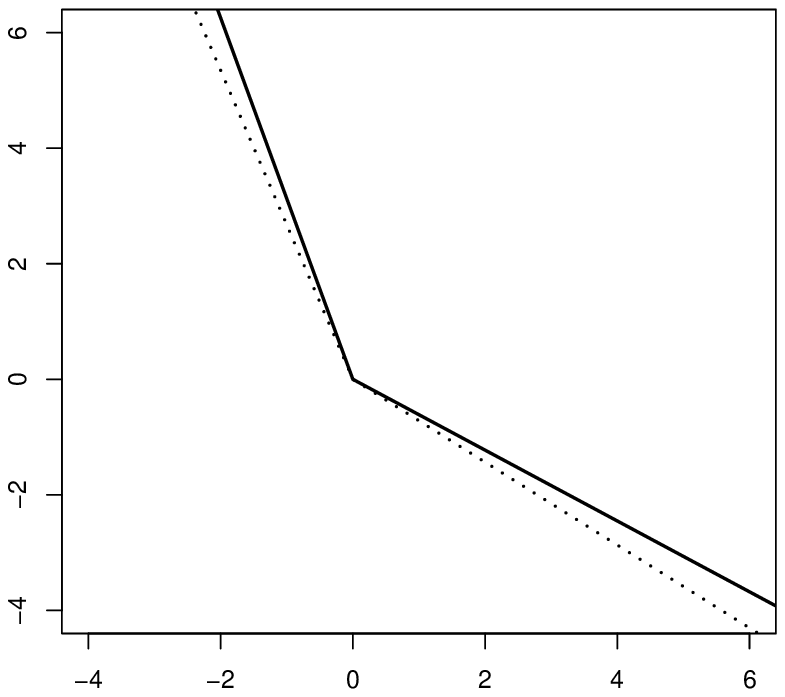}
    \caption{The boundaries of two cones $C_1$ and $C_2$ such that
      $C_1\subset\rhos(\K)\subset C_2$.
      \label{fig:pure-half}}
  \end{figure}

  In order to approximate the risk of $X+\K$, we take a sample of
  $n$ observations of $(X,\pi)$. Denote by $(X^{(n)},\pi^{(n)})$ a random
  vector whose distribution is the empirical distribution of the
  sample of $(X,\pi)$ and by $\K^{(n)}$ a random half-space with normal
  $(\pi^{(n)},1)$. Note that the empirical distribution of the exchange
  rate is bounded away from the origin and infinity.

  Figure~\ref{fig:srisknormhalf}(a) shows a sample of $n=1000$
  observations of a standard bivariate normal distribution and an
  approximation to the true value of $\rhos(X^{(n)}+\K^{(n)})$
  obtained as the sum of the convex hull of the risks of the
  selections of $X^{(n)}+\K^{(n)}$ described in
  (\ref{eq:ex_frictionless}) and $C_1$ (solid line). The boundary of
  $\vecrisk(X^{(n)})+C_2$ is shown as dotted line. The tip of the
  solid cone is approximately at
  $(\mathrm{ES}_{0.05}(\xi_1),\mathrm{ES}_{0.05}(\xi_2))=(1.1085,0.829)$,
  where $\xi=X+\eta_*$ is the selection from \eqref{eq:sel-p}. Since
  the tip of the cone is the tangent point to
  $\rhos(X^{(n)}+\K^{(n)})$ in direction $-(0,\pi_0)$, it yields the
  minimal initial capital requirement of $1.661$ units of the first
  currency, as mentioned in Example~\ref{ex:intro}.

  \begin{figure}[h!]
    \begin{center}
      \begin{tabular}{cc}
        \subfigure[]{\includegraphics[scale=.45,angle=-90]{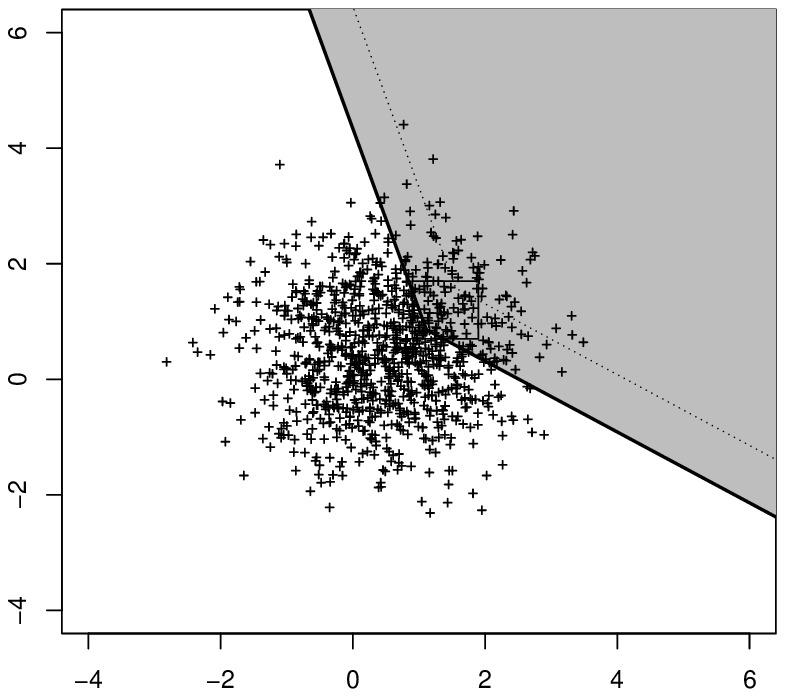}}
        &
        \subfigure[]{\includegraphics[scale=.45,angle=-90]{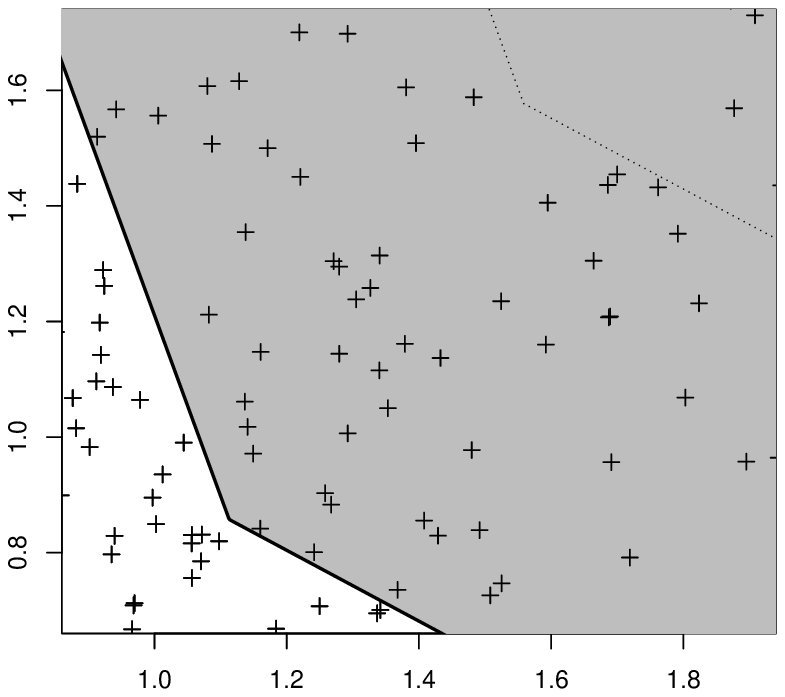}}
      \end{tabular}
      \caption{(a) An approximation to the selection risk measure and
        its upper bound for a normal sample in a random frictionless
        exchange case; (b) enlarged part of the
        plot.\label{fig:srisknormhalf}}
    \end{center}
  \end{figure}
\end{example}

\begin{example}
  \label{ex:restricted-liquidity}
  Consider the restricted liquidity situation from
  Example~\ref{ex:illiquid}. Let $\X=X+(\K\cap ((1,1)+\R_-^2))$, where
  $X$ follows a bivariate standard normal distribution and $\K$ is the
  half-plane as in Examples~\ref{ex:intro}
  and~\ref{ex:two-currencies}. Assume that $\vecrisk=({\rm
    ES}_{0.05},{\rm ES}_{0.05})$. Transactions that diminish the risk
  of $\X$ can be constructed by projecting $X=(X_1,X_2)$ onto the
  boundary line to half-plane $\K$ with normal $(\pi,1)$ and then
  subtracting the projection from $X$. If the obtained point lies out
  of the line segment with end-points $X+(1,-\pi)$ and $X+(-1/\pi,1)$,
  we take the nearest of the two end-points. This leads to a selection
  $\eta$ of $\X$ given by \eqref{eq:eta-intro}. 
  Other relevant selections of $\X$ are $X+(1,-\pi)$ and
  $X+(-1/\pi,1)$.

  In order to approximate the risk of $\X$, we take a sample of
  $n=1000$ observations of $(X,\pi)$. Denote by $(X^{(n)},\pi^{(n)})$
  a random vector whose distribution is the empirical distribution of
  the sample, by $\K^{(n)}$ a random half-space with normal
  $(\pi^{(n)},1)$, and let $\X^{(n)}=X^{(n)}+(\K^{(n)}\cap ((1,1)+\R_-^2))$.

  Figure~\ref{fig:srisknormhalflim}(a) shows a sample of $n=1000$
  observations of $X$ and an approximation to the true value of
  $\rhos(\X^{(n)})$ obtained by calculating risks of all convex
  combinations of the selection $\eta$ of $\X_n$ defined by
  \eqref{eq:eta-intro} and the selections $X^{(n)}+(1,-\pi^{(n)})$,
  $X^{(n)}+(-1/\pi^{(n)},1)$. The shaded region with boundary plotted
  as a dotted line is $\E\check{\X}^{(n)}$, the reflected selection
  expectation of $\X^{(n)}$. The boundary of $\rhoselnot(\X^{(n)})$ is
  plotted as a dashed line, and the boundary of
  $\vecrisk(X^{(n)})+\rhos(\K^{(n)}\cap(\R_-^2+(1,1)))$ as a dash-dot
  line.  The point $({\rm ES}_{0.05}(\eta_1),{\rm
    ES}_{0.05}(\eta_2))=(1.125,0.915)$ belongs to the solid line and
  yields the capital reserves of $1.735$ mentioned in
  Example~\ref{ex:intro}. 

  \begin{figure}[h!]
    \begin{center}
      \begin{tabular}{cc}
        \subfigure[]{\includegraphics[scale=.45,angle=-90]{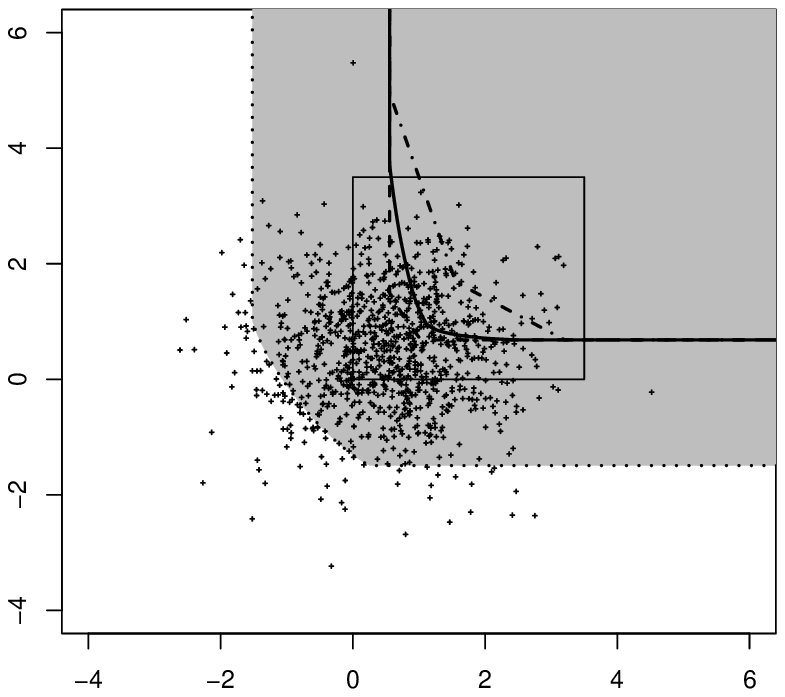}}
        &
        \subfigure[]{\includegraphics[scale=.45,angle=-90]{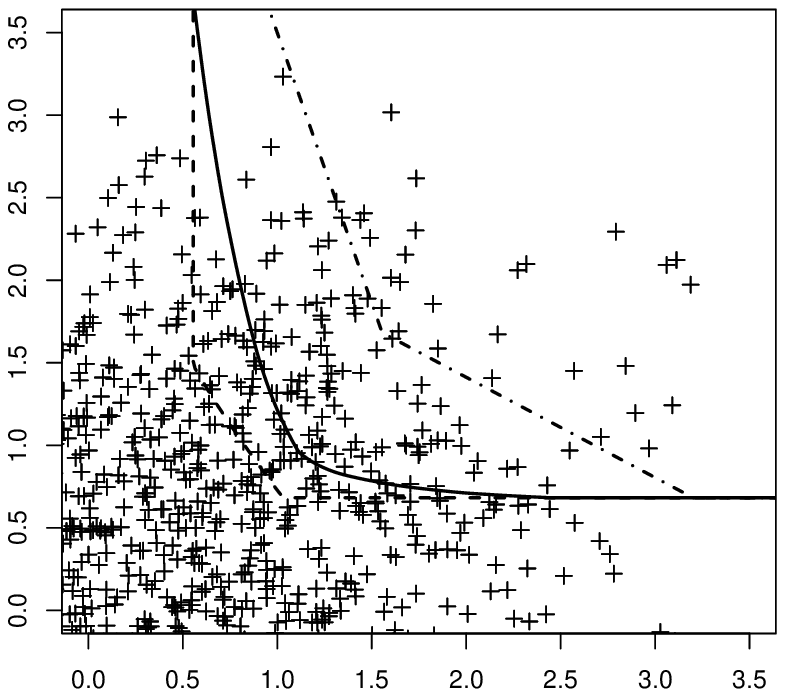}}
      \end{tabular}
      \caption{(a) An approximation to the selection risk measure and
        bounds for a normal sample in a random frictionless exchange
        case with restricted liquidity; (b) enlarged part of the
        plot.\label{fig:srisknormhalflim}}
    \end{center}
  \end{figure}
\end{example}


\newcommand{\noopsort}[1]{} \newcommand{\printfirst}[2]{#1}
  \newcommand{\singleletter}[1]{#1} \newcommand{\switchargs}[2]{#2#1}

\end{document}